\RequirePackage[backend=bibtex, style=ieee, citestyle=numeric-verb, sorting=nyt, natbib=true]{biblatex}
\bibliography{references.bib}

\documentclass{amsart}

\usepackage{graphicx}
\usepackage{amsfonts}
\usepackage{amsmath}
\usepackage{mathtools}
\usepackage{amssymb}
\usepackage{enumerate}
\usepackage{bm}
\usepackage{bbm}
\usepackage{amsaddr}

\newcommand{\calgebra}{\mathfrak{A}}
\newcommand{\nalgebra}{\mathfrak{M}}

\newcommand{\nanalytic}{\mathfrak{M}_\mathcal{A}}
\newcommand{\ex}{\mathfrak{E}}
\newcommand{\hilbert}{\mathcal{H}}
\newcommand{\iu}{i\mkern1mu}
\newcommand{\Dom}[1]{\mathcal{D}\left(#1\right)}
\newcommand{\ie}{\textit{i}.\textit{e}.}
\newcommand{\eg}{\textit{e}.\textit{g}.}

\newcommand{\ip}[2]{\left\langle #1 , #2 \right \rangle}

\makeatletter
\newcommand*{\defeq}{\mathrel{\rlap{%
			\raisebox{0.3ex}{$\m@th\cdot$}}%
		\raisebox{-0.3ex}{$\m@th\cdot$}}%
	=}
\makeatother

\let\Re\relax
\newcommand{\Re}[1]{\operatorname{Re}\left(#1\right)}
\let\Im\relax
\newcommand{\Im}[1]{\operatorname{Im}\left(#1\right)}
\newcommand{\sgn}[1]{\operatorname{sgn}\left(#1\right)}

\theoremstyle{definition}
\newtheorem{definition}{Definition}[section]
\newtheorem{remark}[definition]{Remark}
\newtheorem{example}[definition]{Example}
\newtheorem{notation}[definition]{Notation}
\theoremstyle{plain}
\newtheorem{theorem}[definition]{Theorem}
\newtheorem{lemma}[definition]{Lemma}
\newtheorem{corollary}[definition]{Corollary}
\newtheorem{proposition}[definition]{Proposition}

\begin{document}
	
	\title{Perturbations of KMS States and Noncommutative $L_p$-spaces}
	\thanks{This study was financed in part by the Coordenação de Aperfeiçoamento de Pessoal de Nível Superior - Brasil (CAPES) - Finance Code 001.}
	\author{Ricardo Correa da Silva}
	\address{Institute of Physics - University of S\~ao Paulo}
	\email{ricardo.correa.silva@usp.br}

	\maketitle
	
	\begin{abstract}
		We extend the theory of perturbations of KMS states to a class of unbounded perturbations using noncommutative $L_p$-spaces. We also prove certain stability of the domain of the Modular Operator associated with a \mbox{$\|\cdot\|_p$}-continuous state. This allows us to define an analytic multiple-time KMS condition and to obtain its analyticity together with some bounds to its norm.
		
	\end{abstract}

	\section*{Acknowledgements}
	
	We are grateful to Christian J\"akel and Jo\~ao Barata for all fruitful discussions during the development of this work. This study was financed in part by the Coordenação de Aperfeiçoamento de Pessoal de Nível Superior - Brasil (CAPES) - Finance Code 001.
	
	\section{Introduction}
	\label{intro}
	
	The problem in obtaining a KMS state for a perturbed Hamiltonian is very important in several areas of Physics. This problem was solved by Araki in \cite{Araki73} when the perturbation is bounded, but, since it is often not the case for perturbations of physical interest, a similar result including unbounded perturbations is desired.
	
	The ``expansionals'' plays a very important role in Araki's perturbation theory. In fact, Araki named the operator in Definition \ref{DysonSeries} expansional because if was so called by Fujiwara. Its importance relies on its relation with the Dyson's series.
	
	One of the most interesting properties for us is that \begin{equation}
	\label{expansional_exponential}
	e^{\iu t(A+B)}e^{-\iu tB}\xi= Exp_l\left(\int_0^t;e^{isB}Ae^{-isB} ds\right)\xi \, , \qquad \xi \in \hilbert,\end{equation}
	where in $A$ and $B$ are bounded operators in the von Neumann algebra $\nalgebra$. This formula can be extended for unbounded operators $A$ and $B$. In fact, this formula holds for vectors $\xi \in G(B)$, the set of geometric vectors with respect to $B$, which are defined as the vectors with the property that there exists a positive constant $M_\xi$ such that $\|B^n\xi\|\leq M_\xi^n \|\xi\|$ if the operators $A$ and $B$ are satisfies $A G(B)\subset G(B)$ and there exists $k\in\mathbb{R}$ such that, for each $t\in\mathbb{R}$, $\tau_t(A)=\left(e^{itB}Ae^{-itB}\right)$ is a bounded operator with $\|\tau_t(A)\|\leq e^{k|t|}$. It is important to notice the physical interpretation that the operator in the left-hand side takes the vector back by the dynamics defined by $B$ and evolves it by the dynamics defined by $A+B$. It is also possible to obtain a complex version of this formula, as can be seen in \cite[Section 1.15 - 1.17]{sakai91}. It is important to stress that we are most interested in the case the imaginary part of the exponent equals $\frac{1}{2}$, since this is the case in which we obtain the vector that represents the KMS state to the perturbed dynamics as can be checked in \cite[Equation 1.10]{Araki73}.
	
	After Araki's original work, some improvements have been done. One of them, due to Sakai, which extends the theory for bounded perturbations bounded from below. Sakai also developed an approach to the problem using derivations in $C^\ast$-algebras. Finally, in \cite{Derezinski03}, an extension of the theory is presented where unbounded perturbations may be included as \hyphenation{per-tur-ba-tions}perturbations.
	
	In this paper, we will present a different approach to Araki's Perturbation Theory using noncommutative $L_p$-spaces and which includes unbounded perturbations. It is important to notice that noncommutative $L_p$-spaces have been successfully used in Linear Response and Constructive Quantum Field Theory, see \cite{nittis2016} and \cite{Bru2017}.
	
	The main results are Theorem \ref{TR0}, Theorem \ref{TR1} and Corollary \ref{CR1}.
	
	\section{Background}
	\label{secBackground}
	The aim of this section is to fix the notation and to present the basis needed to better understand our results. For the reader interested in more details we refer to \cite{RCS18.2}.
	
	In this entire work we denote by $\calgebra$ a $C^\ast$-algebra and by $\nalgebra$ a von Neumann algebra. In addition, $\hilbert$ denotes a Hilbert space and often we will suppose $\calgebra, \nalgebra \subset B(\hilbert)$, \ie, the algebras will be thought as concrete ones. 
	
	\subsection{Modular Theory}
	\label{subsecModular}
	
	This section is devoted to present the definitions and main properties of the modular operator and the modular conjugation. This topic is a standard subject and can be found in classical books \eg \, \cite{Bratteli1}, \cite{KR83} and \cite{Takesaki2002} or even in \cite{Araki74}.
	
	Let us now define two operators in $\nalgebra$, which will give rise to the operators that give name to this section. For the cyclic and separating vector $\Omega$, define the anti-linear operators:
	\begin{center}
		\begin{minipage}[c]{5cm}
			$$\begin{aligned}
			S_0 :	&\{A\Omega \in \hilbert \ | A\in \nalgebra\} 		&\to		& \hspace{0.4cm} \hilbert\\
			&\hspace{1.2cm} A\Omega 								&\mapsto	&\hspace{0.2cm} A^\ast\Omega
			\end{aligned}$$
		\end{minipage}\hspace{0.3cm},\hspace{0.4cm}
		\begin{minipage}[c]{5cm}
			$$\begin{aligned}
			F_0 :	& \{A^\prime\Omega \in \hilbert \ | A^\prime\in \nalgebra^\prime\} 		&\to		& \hspace{0.4cm} \hilbert \\
			& \hspace{1.2cm} A^\prime\Omega 								&\mapsto	&\hspace{0.2cm}A^{\prime\ast}\Omega
			\end{aligned}$$
		\end{minipage}
	\end{center}
	Note that the domains of the operators are dense subspaces.
	It is a standard result that the operators $S_0$ and $F_0$ are closable operators. Moreover, $S_0^\ast=\overline{F_0}$ and $F_0^\ast=\overline{S_0}$.	We will denote $S=\overline{S_0}$ and $F=\overline{F_0}$.
	
	An important point to stress now is that we omitted the dependence on $\Omega$ to keep the notation clean, but we will mention it in the following.
	
	Moreover, even though $S$ is not a bijection, it is injective and we will write $S^{-1}$ (which is equal to $S$) to denote its inverse over its range. The same holds for $\Delta_\Omega$, which will be defined soon.
	
	\begin{definition}
		We denote by $J_\Omega$ and $\Delta_\Omega$ the unique anti-linear partial isometry and positive operator, respectively, in the polar decomposition of $S$, \ie,   $S=J_\Omega\Delta_\Omega^{\frac{1}{2}}$. $J_\Omega$ is called the modular conjugation and $\Delta_\Omega$ is called the modular operator. 
	\end{definition}
	
	Note that the existence and uniqueness of these operators are stated in the Polar Decomposition Theorem.
	
	Several properties hold for the modular operator, we refer to \cite{Araki74}, \cite{Bratteli1} and \cite{RCS18} for the reader interested in this subject.
	
	One of the most important results in Modular Theory is the Tomita-Takesaki Theorem, which is extremely significant to both Physics and Mathematics. The proof and applications of this theorem can be found in \cite{Takesaki2002} and \cite[Theorem 2.5.14]{Bratteli1}. One of the consequences of this theorem is that, for each fixed $t\in\mathbb{R}$, $A\xmapsto{\tau^\Omega_t}\Delta_\Omega^{\iu t}A\Delta_\Omega^{-\iu t}$ defines an isometry of the algebra. Hence, $\left\{\tau^\Omega_t \right\}_{t\in\mathbb{R}}$ is a one-parameter group of isometries.
	
	\begin{definition}[Modular Automorphism Group]
		Let $\nalgebra$ be a von Neumann algebra with cyclic and separating vector $\Omega$ and let $\Delta_\Omega$ be the associated modular operator. For each $t\in\mathbb{R}$, define the isometry $\tau^\Omega:\nalgebra \to \nalgebra$ by $\tau^\Omega_t(A)=\Delta_\Omega^{\iu t}A\Delta_\Omega^{-\iu t}$. We call the modular automorphism group\index{modular! automorphism group} the one-parameter group $\left\{\tau^\Omega_t \right\}_{t\in\mathbb{R}}$.
	\end{definition}
	
	\begin{notation}
		We will denote the modular automorphism group with respect to a cyclic and separating vector $\Omega$ by $\left\{\tau^\Omega_t \right\}_{t\in\mathbb{R}}$. In addition, given a faithful normal semifinite weight $\phi$ on a von Neumann algebra, we will denote $\left\{\tau^\phi_t \right\}_{t\in\mathbb{R}}$ the modular automorphism group with respect to the cyclic and separating vector obtained in the GNS-construction.
	\end{notation}
	
	The last comment we would like to add in this section is that in Relativistic Quantum Field Theory, due to the Reeh-Schlieder Theorem, it is possible to obtain a modular operator for the algebra of local observables using the vacuum state, see \cite{Araki99} for more details.
	
	\subsection{KMS States and Dynamical Systems}
	\label{subsecKMS}
	This section is devoted to establishing the basic concepts and notation about Dynamical Systems in the Operator Algebras context needed in this article.
	
	\begin{definition}[$W^\ast$-Dynamical System]
		\index{dynamical system! $W^\ast$}
		A $W^\ast$-dynamical system $(\nalgebra,G,\alpha)$ consists of a von Neumann algebra $\nalgebra$, a locally compact group $G$ and a weakly continuous homomorphism $\alpha$ of $G$ in $Aut(\nalgebra)$.
	\end{definition}
	
	In particular, we denote by $(\nalgebra,\alpha)$ the $W^\ast$-dynamical system with $\alpha$ a one-parameter group, $\mathbb{R} \ni t \mapsto \alpha_t \in Aut(\nalgebra)$.
	
	A general definition of KMS states can be found in any textbook of Operator Algebras such as \cite{Bratteli2}, \cite{KR86}, and \cite{Takesaki2002}. For a discussion in the context of equilibrium states in the thermodynamic limit we suggest reference \cite{haag67}. We will, present this definition for completeness
	
	\begin{definition}
		\label{KMSequi}
		Let $(\nalgebra, \tau)$ be a $W^\ast$-dynamical system, $\beta \in \mathbb{R}$. A normal state $\omega$ over $\nalgebra$ is said to be a $(\tau,\beta)$-KMS state, if, for any $A,B \in \nalgebra$, there exists a complex function $F_{A,B}$ which is analytic in $\mathcal{D}_\beta=\left\{z\in \mathbb{C} \mid 0<\sgn{\beta} \Im z<|\beta|\right\}$ and continuous on $\overline{\mathcal{D}_\beta}$ satisfying
		\begin{equation}
		\label{eqKMS3}
		\begin{aligned}
		F_{A,B}(t) &=& \omega(A\tau_t(B)) \ \forall t\in \mathbb{R}, \\
		F_{A,B}(t+\iu \beta)& =& \omega(\tau_t(B)A) \ \forall t\in \mathbb{R}. \\
		\end{aligned}
		\end{equation}
		
	\end{definition}

	We finish this section mentioning a result of the most importance. KMS states ``survive the thermodynamical limit'', which means that under topology that has physical significance in this mathematical description and under the adequate convergence and continuity hypothesis the limit of KMS states is again a KMS state.

	\subsection{Expansionals}
	
	A extensive study of expansionals can be found in \cite{Araki73.2}.
	
	\begin{definition}
		\label{DysonSeries}\index{expansional}
		Let $\nalgebra$ be a von Neumann algebra, $t\mapsto A(t)\in\nalgebra$ a strongly continuous function such that $\displaystyle \sup_{0\leq t\leq T}\|A(t)\|=r_A(T)<\infty$ for all $T\in\mathbb{R}_+$. For each $t\in\mathbb{R}_+$ define
		$$\begin{aligned}
		Exp_r\left(\int_0^t;A(s) ds\right)	&=\sum_{n=0}^{\infty}{\int_{0}^{t} dt_1\ldots \int_{0}^{t_{n-1}}dt_{n}A(t_n)\ldots A(t_1)}; \\
		Exp_l\left(\int_0^t;A(s) ds\right)	&=\sum_{n=0}^{\infty}{\int_{0}^{t} dt_1\ldots \int_{0}^{t_{n-1}}dt_{n} A(t_1)\ldots A(t_n)}; \\
		\end{aligned}$$
		where the term for $n=0$ is the identity.
	\end{definition}
	
	Note that these operators are well defined since $\|A(t_i)\|\leq r_A(t)$ for every $1\leq i\leq n$ and ${\int_{0}^{t} dt_1\ldots \int_{0}^{t_{n-1}}dt_{n}=\frac{t^n}{n!}}$. Thus the series converge absolutely (and uniformly over compact sets).
	
	It is important to mention that these operators are basically the Dyson series. In fact, if given $(t_1,\ldots, t_n) \in \mathbb{R}^n$, we set a permutation $\sigma:\{1,\ldots,n\}\to\{1,\ldots,n\}$ such that $t_{\sigma(n)}\leq t_{\sigma(n-1)}\leq \ldots t_{\sigma(1)}$, and we define the operators $T,\widetilde{T}:\nalgebra\to\nalgebra$ by
	$$\begin{aligned}
	T\left(A(t_1)\ldots A(t_n)\right)&=A(t_{\sigma(1)})\ldots A(t_{\sigma(n)});\\
	\widetilde{T}\left(A(t_1)\ldots A(t_n)\right)&=A(t_{\sigma(n)})\ldots A(t_{\sigma(1)});
	\end{aligned}$$
	then, 
	$$\begin{aligned}
	Exp_r\left(\int_0^t;A(s) ds\right)	&=\sum_{n=0}^{\infty}{\int_{0}^{t} dt_1\ldots \int_{0}^{t}dt_{n}\frac{T\left(A(t_n)\ldots A(t_1)\right)}{n!}}; \\
	Exp_l\left(\int_0^t;A(s) ds\right)	&=\sum_{n=0}^{\infty}{\int_{0}^{t} dt_1\ldots \int_{0}^{t}dt_{n} \frac{\widetilde{T}\left(A(t_n)\ldots A(t_1)\right)}{n!}}. \\
	\end{aligned}$$
	
	The following proposition states some interesting properties of expansionals, for example the cocycle property. Equation \eqref{expansional_exponential} also can be obtained from this proposition. For the proof and more details see \cite[Proposition 2, 3, 4 and 5]{Araki73.2} or \cite{RCS18.2}.
	
	\begin{proposition}
		\label{exponentials}
		Let $\nalgebra$ be a von Neumann algebra, $t\mapsto A(t)\in\nalgebra$ a strong-continuous function such that $\displaystyle \sup_{0\leq t\leq T}\|A(t)\|=r_A(T)<\infty$ for all $T\in\mathbb{R}_+$. Then, the following properties hold:
		\begin{enumerate}[(i)]
			\item $\displaystyle \frac{d}{dt}{Exp_r\left(\int_0^t;A(s) ds\right)}=Exp_r\left(\int_0^t;A(s) ds\right)A(t)$;\\
			$\displaystyle \frac{d}{dt}{Exp_l\left(\int_0^t;A(s) ds\right)}=A(t)Exp_l\left(\int_0^t;A(s) ds\right)$;
			\item $\displaystyle Exp_l\left(\int_0^t;-A(s) ds\right) Exp_r\left(\int_0^t;A(s) ds\right)=\mathbbm{1}$;
			\item $\displaystyle Exp_r\left(\int_0^t;A(s) ds\right) Exp_l\left(\int_0^t;-A(s) ds\right)=\mathbbm{1}$;
			\item $\displaystyle Exp_r\left(\int_0^t;A(s) ds\right) Exp_r\left(\int_0^{t^\prime};A(s+t) ds\right)=Exp_r\left(\int_0^{t+t^\prime};A(s) ds\right)$;
			
			$\displaystyle Exp_l\left(\int_0^{t^\prime};A(s+t) ds\right) Exp_l\left(\int_0^t;A(s) ds\right)=Exp_l\left(\int_0^{t+t^\prime};A(s) ds\right)$.
		\end{enumerate}
	\end{proposition} 
	
	In addition, it is known, \cite[Proposition 4.3]{Araki73}, that the modular automorphism groups $\{\tau_t^\phi\}_{t\in\mathbb{R}}$ and $\{\tau_t^\psi\}_{t\in\mathbb{R}}$ of two states $\phi$ and $\psi$, respectively, when the relative Hamiltonian $Q$ exists, can be related as follows:
	$$\begin{aligned}
	u^{\phi\psi}_t&=Exp_r\left(\int_0^t;-\iu \tau^\psi_s(Q) ds\right)\\
	\hat{u}^{\phi\psi}_t&=Exp_l\left(\int_0^t;\iu \tau^\psi_s(Q) ds\right).\\
	\end{aligned}$$
	$$\begin{aligned}
	\left(u^{\phi\psi}_t\right)^\ast&=\hat{u}^{\phi\psi}_t\\
	u^{\phi\psi}_t \hat{u}^{\phi\psi}_t&=\hat{u}^{\phi\psi}_t u^{\phi\psi}_t=\mathbbm{1}\\
	u^{\phi\psi}_t \tau^\psi_t(A)&=\tau^\phi_t(A) \hat{u}^{\phi\psi}_t \, , \qquad A\in \nalgebra.
	\end{aligned}$$
	Moreover, the vectors $\Phi$ and $\Psi$ representing, respectively, the states $\phi$ and $\psi$ are related by
	
	$$\Phi=\sum_{i=0}^\infty (-1)^n \int_0^\frac{1}{2} \int_0^{t_1}\ldots \int_0^{t_{n-1}}\Delta_\Psi^{t_n}Q\Delta_\Psi^{t_{n-1}-t_n}Q\cdots \Delta_\Psi^{t_1-t_2}Q\Psi.$$
	\subsection{Noncommutative $L_p$-Spaces}
	\label{subsecNCLp}
	
	Noncommutative $L_p$-spaces are analogous to the Banach spaces of the $p$-integrable functions with respect to a measure. The study of these spaces goes back to the works of Segal, \cite{Segal53}, and Dixmier, \cite{Dixmier53}, which depend on the existence of a normal faithful semifinite trace. It was just 25 years later that Haagerup in \cite{haagerup79} proposed a generalization of the Segal-Dixmier $L_p$-spaces which included the type III von Neumann algebras. As a consequence of \cite{hilsum81}, which answers to a question on spacial derivatives raised by A. Connes, that Araki and Masuda could propose a definition, equivalent to that one proposed by Haagerup, of noncommutative $L_p$-space based just in the Hilbert space of a concrete von Neumann algebra.
	
	It is interesting to notice that noncommutative $L_p$-spaces are appearing more frequently as the best framework to describe some physical situations. See, for example \cite{germinet2005}, \cite{Bru2017}, \cite{germinet2011}, and \cite{nittis2016}.
	
	Our interest in these spaces for a class of perturbations can be justified on two well known facts for classical $L_p$-spaces that still hold in the \hyphenation{non-com-mu-ta-tive}noncommutative case: they admit unbounded functions (operators) and have a useful H\"older duality property.
	
	In this section, we will present the useful theory of noncommutative measures whith respect to a normal faithful semifinite trace on a von Neumann algebra, which is the basis for the Segal-Dixmier noncommutative $L_p$-spaces (and, by the way, for noncommutative geometry). We use \cite{terp81} very often in here.
	
	
	Henceforth, we will denote by $\tau$ a trace, meaning a normal faithful semifinite trace. It is important to note that supposing the existence of such a trace restricts our options of algebras to the semifinite ones (not type III).

	Given a von Neumann algebra $\nalgebra\subset B(\hilbert)$, we say that a closed dense defined linear operator $A:\Dom{A} \to \hilbert$ is affiliated to $\nalgebra$ if, for every unitary operator $U\in\nalgebra^\prime$, $UAU^\ast=A$. We denote that an operator is affiliated to $\nalgebra$  by $A\eta\nalgebra$ and the set of all affiliated operators by $\nalgebra_\eta$.
		
	\begin{definition}
		Let $\nalgebra$ be a von Neumann algebra, $\tau$ a normal faithful semifinite trace, and $\varepsilon, \delta>0$. Define $$D(\varepsilon,\delta)=\left\{A\in \nalgebra_\eta \ \middle| \ \begin{aligned} &\exists p\in \nalgebra_p \textrm{ such that } \\ &p\hilbert\subset\Dom{A}, \|Ap\|\leq\varepsilon \textrm{ and } \tau(\mathbbm{1}-p)\leq\delta\end{aligned}\right\}.$$
	\end{definition}
	
	\begin{proposition}
		\label{pn->1} \index{$\tau$-! dense}
		Let $\nalgebra$ be a von Neumann algebra and $\tau$ a normal faithful semifinite trace. A subspace $V\subset\hilbert$ is $\tau$-dense if, and only if, there exists an increasing sequence of projections $(p_n)_{n\in\mathbb{N}}\subset\nalgebra_p$ such that $p_n\to \mathbbm{1}$ and $\tau(\mathbbm{1}-p_n)\to 0$ and $\displaystyle \bigcup_{n\in\mathbb{N}} p_n\hilbert\subset V$.
	\end{proposition}

	\begin{corollary}
		\label{intersec2taudense}
		Let $V_1,V_2 \subset\hilbert$ be $\tau$-dense subspaces. Then $V_1\cap V_2$ is $\tau$-dense. 	
	\end{corollary}
	
	\begin{definition}
		\index{$\tau$-! measurable}
		Let $\nalgebra$ be a von Neumann algebra and $\tau$ be a normal faithful semifinite trace. A closed (densely defined) operator $A\in \nalgebra_\eta$ is said $\tau$-measurable if $\Dom{A}$ is $\tau$-dense. We denote by $\nalgebra_\tau$ the set of all $\tau$-measurable operators.
	\end{definition}
	
	Notice that by the previous proposition, if $A$ is a $\tau$-measurable operator and $B$ extends $A$, we must have $A=B$. This, in turn, implies that a $\tau$-measurable symmetric operator is self-adjoint.
	
	\begin{definition}
		\label{defpremeasurable}
		Let $\nalgebra$ be a von Neumann algebra and $\tau$ be a normal faithful \hyphenation{semi-fi-nite}semifinite trace. An operator $A\eta\nalgebra$ is said $\tau$-premeasurable if, $\forall \delta>0$, there exists $p\in\nalgebra_p$ such that $p\hilbert \subset \Dom{A}$, $\|Ap\|<\infty$ and $\tau(\mathbbm{1}-p)\leq \delta$.  
	\end{definition}
	
	An equivalent way to define a $\tau$-premeasurable operator relies on $D(\varepsilon,\delta)$: $A$ is $\tau$-premeasurable if, and only if, $\forall \delta>0$, there exists $\varepsilon>0$ such that $A\in D(\varepsilon,\delta)$.
	
	Another interesting thing to notice is that a $\tau$-premeasurable operator is densely defined since $D(A)$ must be $\tau$-dense.
	
	\begin{proposition}
		\label{equivalencetaumeasurability}
		Let $\nalgebra$ be a von Neumann algebra, $\tau$ a normal faithful semifinite trace, $A \eta \nalgebra$ a closed densely defined operator, and $\left\{E_{(\lambda,\infty)}\right\}_{\lambda\in\mathbb{R}_+}$ the spectral resolution of $|A|$. The following are equivalent:
		\begin{enumerate}[(i)]
			\item $A$ is $\tau$-measurable;
			\item $|A|$ is $\tau$-measurable;
			\item $\forall \delta>0 \ \exists \varepsilon>0$ such that $A\in D(\varepsilon,\delta)$;
			\item $\forall \delta>0 \ \exists \varepsilon>0$ such that $\tau\left(E_{(\varepsilon,\infty)}\right)<\delta$;
			\item $\displaystyle \lim_{\lambda\to\infty}\tau\left(E_{(\lambda,\infty)}\right)=0$;
			\item $\exists \lambda_0>0$ such that $\tau\left(E_{(\lambda_0,\infty)}\right)<\infty$.
		\end{enumerate}
	\end{proposition}
	\begin{proposition}
		$\nalgebra_\tau$ provided with the usual scalar operations and involution, and the following vector operations, is a $\ast$-algebra:
		\begin{enumerate}[(i)]
			\item $A\bm{+}B=\overline{A+B}$;
			\item $A\bm{\times}B=\overline{AB}$.
		\end{enumerate}
	\end{proposition}
	\begin{proposition}
		\label{measurablealgebra}
		$\nalgebra_\tau$ is a complete Hausdorff topological $\ast$-algebra with respect to the topology generated by the system of neighbourhoods of zero $$\left\{\nalgebra_\tau \cap D(\varepsilon,\delta)\right\}_{\varepsilon>0, \delta>0}.$$ Furthermore, $\nalgebra$ is dense in $\nalgebra_\tau$ in this topology. We will denote the balanced absorbing neighbourhood of zero by $N(\varepsilon,\delta)=\nalgebra_\tau \cap D(\varepsilon,\delta)$.
	\end{proposition}
	
	It interesting to notice that analyticity pervades almost every subject in von Neumann algebras. As a consequence of linearity and normality of the trace, we can use Functional Calculus and Spectral Theory to take advantage of the well known rigid behaviour of analytic functions to prove the aforesaid inequalities. The details of the proofs can be found in \cite{RCS18}, we also refer to \cite{Ruscai72}.
	\begin{lemma}
		\label{normtraceinequality}
		Let $\nalgebra$ be a von Neumann algebra, $\tau$ a normal faithful semifinite trace on $\nalgebra$, $A\in \nalgebra$ and $B\in \mathfrak{M}_\tau$. Then
		$$\left|\tau(AB)\right|\leq\tau(|AB|)\leq \|A\|\tau(|B|). $$
	\end{lemma}

	\begin{theorem}[H\"older Inequality\index{inequality! H\"older}]
		\label{gholder}
		Let $\nalgebra$ be a von Neumann algebra and $\tau$ a normal faithful semifinite trace in $\nalgebra$. Let also $A_i\in\nalgebra$, $i=1,\dots, k$ and $\displaystyle \sum_{i=1}^{k} p_i>1$ such that $\displaystyle \sum_{i=1}^{k}\frac{1}{p_i}=1$, then
		$$\tau\left(\left|\prod_{i=1}^{k}A_i\right|\right)\leq \prod_{i=1}^{k}\tau(|A_i|^{p_i})^\frac{1}{p_i}.$$
	\end{theorem}
	
	The reader should keep in mind that H\"older's inequality is a very interesting result to us, since it says something regarding the trace of a product and this is the case in Dyson's series. Nevertheless, it is used in the proof of the Minkowski Inequality which is imperative to define a normed vector space.
	
	\begin{theorem}[Minkowski's Inequality\index{inequality! Minkowski}]
		\label{minkowski}
		Let $\nalgebra$ be a von Neumann algebra, $\tau$ a normal faithful semifinite trace in $\nalgebra$, and $p,\, q>1$ such that $\frac{1}{p}+\frac{1}{q}=1$. Then
		\begin{enumerate}[(i)]
			\item for every $A\in \nalgebra$, $\displaystyle \tau(|A|^p)^\frac{1}{p}=\sup\left\{\left|\tau(AB)\right| \ \middle | \ B\in \nalgebra,  \tau\left(|B|^q\right)\leq 1\right\};$
			\item for every $A,B \in \nalgebra$, $\displaystyle \|A+B\|_p\leq \|A\|_p+\|B\|_p$.
		\end{enumerate}
	\end{theorem}
	
	Together, Theorem \ref{gholder} and Theorem \ref{minkowski} provide us with another generalization of H\"older's inequality. This inequality is obvious in the \hyphenation{com-mu-ta-tive}commutative case, but not in the noncommutative case.
	
	\begin{corollary}[H\"older Inequality\index{inequality! H\"older}]
		\label{g2holder}
		Let $\nalgebra$ be a von Neumann algebra and $\tau$ a normal faithful semifinite trace in $\nalgebra$, let also $A,B\in\nalgebra$ and $p,\, q>1$ such that $\frac{1}{p}+\frac{1}{q}=\frac{1}{r}$. Then
		$$\tau(|AB|^r)^\frac{1}{r}\leq \tau(|A|^p)^\frac{1}{p}\tau(|B|^q)^\frac{1}{q}.$$
	\end{corollary}
	
	\begin{definition}\index{noncommutative $L_p$-space! Segal-Dixmier}
		Let $\nalgebra$ be a von Neumann algebra and $\tau$ a normal, faithful and semifinite trace on $\nalgebra$. We define the noncommutative $L_p$-space, denoted by $L_p(\nalgebra,\tau)$, as the completion of
		$$\left\{A\in \nalgebra \ \middle| \ \tau\left(|A|^p\right)<\infty\right\}$$
		with respect to the norm $\displaystyle\|A\|_p=\tau\left(|A|^p\right)^\frac{1}{p}$.
		
		We also set $L_\infty(\nalgebra, \tau)=\nalgebra$ with $\|A\|_\infty=\|A\|$.
		
	\end{definition}
	
	Now, it is easy to see that, for $p,\, q\geq1$ H\"older conjugated, the H\"older and Minkowski inequalities can be extended to the whole space $L_p(\nalgebra,\tau)$ through an argument of density  and normality of the trace. With this definition, Lemma \ref{normtraceinequality} and Corollary \ref{g2holder}, and Theorem \ref{minkowski} can be expressed as
	$$\begin{aligned}
	\|AB\|_1&\leq \|A\|_p\|B\|_q,\\
	\|A+B\|_p &\leq \|A\|_p+\|B\|_p,\\
	\end{aligned}$$ 
	and this last inequality is a triangular inequality for $\|\cdot\|_p$. It is important to notice that faithfulness guarantees $\|A\|_p=0 \Rightarrow A=0$, however semifiniteness was used only at the very end of Theorem \ref{minkowski} and it is completely irrelevant when talking about noncommutative $L_p$-spaces, since the trace is never infinity on these operators.
	
	It is not our intention in this text to discuss this subject, but notice that if $\tau$ is not semifinite, we can define the noncommutative $L_p$ space to a ``small'' algebra $\overline{\mathfrak{M}_\tau}^{SOT}$.
	
	\begin{theorem}
		\label{dualLp}
		Let $p,\, q\geq 1$ such that $\frac{1}{p}+\frac{1}{q}=1$. Then the function below is an isometric isomorphism:
		
		\begin{center}
			\begin{tabular}{@{\hskip 2pt}c@{\hskip 2pt}c@{\hskip 2pt}c@{\hskip 2pt}c@{\hskip 2pt}c@{\hskip 2pt}c}
				$\Xi$: $L_p(\nalgebra,\tau)$&$\to$& $L_q(\nalgebra,\tau)^\ast$& & \\
				$A$ &$\mapsto$&$\tau_A:$&$L_q(\nalgebra,\tau)$&$\to $&$\mathbb{C}$\\
				& & &$B$&$\mapsto $&$\tau(AB)$.\\
			\end{tabular}
		\end{center}
	\end{theorem}
	
	This last result is the famous identification $L_p(\nalgebra,\tau)^\ast=L_q(\nalgebra,\tau)$ where $p,\, q>1$ are H\"older conjugated.
	
	We do not intend to do a long presentation about the Radon-Nikodym Theorem, to which the next proposition is somewhat related, but we need to write the next result since it is important to understand what is done here.
	
	\begin{proposition}
		\label{unboundedderivativeweight}
		Let $\nalgebra$ be a von Neumann algebra, $\phi$ a faithful normal semifinite weight on $\nalgebra$ and $H\eta \mathfrak{M}_{\tau}^+$.  If $\left(H_i\right)_{i\in I} \in  \mathfrak{M}_{\tau}^+$ is an increasing net such that $H_i \to H$, then
		\begin{equation}
		\label{eq:definitionx1}
		\phi_H(A)\doteq\sup_{i\in I}{\phi\left(H_i^\frac{1}{2}AH_i^\frac{1}{2}\right)}, \quad A\in\nalgebra
		\end{equation}
		defines a normal semifinite weight $\phi_H$ on $\nalgebra$, which is independent of the choice of the net $(H_i)_{i\in I}$ with $H_n \to H$. 
		In addition, $\phi_H$ is faithful if, and only if, $H$ is non-singular.
		
		 Moreover, if $(H_i)_{i\in I}$ is an increasing  net of positive operators affiliated with $\mathfrak{M}_{\tau}^+$ such that $H_i\to H$, then 
		$$\phi_H=\sup_{i\in I}\phi_{H_i}.$$
		
	\end{proposition}
	
	Notice that if $\nalgebra\subset B(\hilbert)$, the conditions on the net in the previous theorem is equivalente to $H_i \xrightarrow{SOT} H$, because of the Vigier's Theorem. Hence, the next equation holds
	\begin{equation}
	\label{eq:calculationX17}
	\phi_H(A)=\lim_{H_i \to H}\phi\left(H_i^\frac{1}{2}AH_i^\frac{1}{2}\right)=\lim_{H_i \to H}\phi\left(H_i A\right), \quad A\in \nalgebra.
	\end{equation}

	\def\typeiii{0}
	\if\typeiii1
	
	Roughly speaking, Haagerup's generalization of noncommuative $L_p$-spaces for general von Neumann algebras uses several identifications between a von Neumann algebra (and other objects related to it) and the crossed product $\nalgebra\rtimes_\alpha \mathbb{R}$, where $\{\tau^{\varphi}_t\}_{t\in\mathbb{R}}$ is the modular automorphism group obtained throughout the faithful normal and semifinite weight $\varphi$. In the following we will present some of these identifications based on our convenience. A detailed presentation can be found in \cite{Haagerup79I}, \cite{Haagerup79II} and \cite{terp81}.
	
	Consider the representations $\pi_\alpha$ of $\nalgebra$ and $\lambda$ of $\mathbb{R}$ both in the Hilbert space $L_2(\mathbb{R},\nalgebra)$ used in the definition of the crossed product, \ie, defined by
	$$\begin{aligned}
	\left(\pi_\alpha(A)(f)\right)(t)&=\alpha^{-1}_t(A) f(t), & A\in \nalgebra, \ & \forall f\in L_2(\mathbb{R},\nalgebra), \ & \forall t\in \mathbb{R},\\
	\left(\lambda(g)f\right)(t)&=f(g^{-1}t), & g\in \mathbb{R}, \ & \forall f\in L_2(\mathbb{R},\nalgebra), \ & \forall t\in \mathbb{R}.\\
	\end{aligned}$$ 
	
	Consider also the dual action $\theta$ of $\mathbb{R}$ in $\nalgebra\rtimes_\alpha \mathbb{R}$, characterized by
	$$\begin{aligned}
	&\theta_s A=A, &\ A\in \pi_\alpha(\nalgebra),\\
	&\theta_s \lambda(t)=e^{-\iu st}\lambda(t). & \ t\in \mathbb{R}.\\
	\end{aligned}$$ 
	
	Notice that by \cite[Lemma 3.6]{Haagerup79I}, we have the following characterization:
	$$\pi_\alpha(\nalgebra)=\left\{A\in \nalgebra\rtimes_\alpha \mathbb{R}  \ \middle| \ \theta_t A=A \ \forall t\in\mathbb{R}\right\}. $$
	
	From now until the end of this section we will identify $\nalgebra$ with $\pi_\alpha(\nalgebra)$.
	
	\begin{lemma}
		\label{Tcrossedproduct}
		Let $\nalgebra$ be a von Neumann algebra. The following properties hold for the operator $T:\left(\nalgebra\rtimes_\alpha \mathbb{R}\right)_+\to \widehat{\nalgebra}_+$, given by
		$$TA=\int_{-\infty}^{\infty}\theta_t(A)dt, \quad A\in \left(\nalgebra\rtimes_\alpha \mathbb{R}\right)_+ \ ,$$
		and characterized by
		$$(\phi,TA)=\int_{-\infty}^{\infty}\phi\circ\theta_t(A)dt, \quad \forall \phi\in\nalgebra_\ast^+:$$
		\begin{enumerate}[(i)]
			\item $T$ is a normal faithful semifinite operator valued weight;
			\item There exists a unique normal faithful semifinite trace $\tau$ on $\nalgebra \rtimes_\alpha \mathbb{R}$ such that $(D\phi\circ T:D\tau)_t=\lambda(t)$ for all $t\in \mathbb{R}$;
			\item The trace $\tau$ satisfies $\tau\circ\theta_t=e^{-t}\tau$ for all $t\in\mathbb{R}$.
		\end{enumerate}
	\end{lemma}
	
	Since we will be interested in the weights on $\nalgebra$, \cite[II, Lemma 1]{terp81} gives us a useful identification
	\begin{lemma}
		\label{identificationofweights}
		The mapping $\phi\mapsto \tilde{\phi}$, where $\tilde{\phi}=\hat{\phi}\circ T$ and $\hat{\phi}$ is a natural extension of $\phi$ to $\widehat{\nalgebra}_+$ (the extended positive part of $\nalgebra$) as described in \cite[Proposition 1.10]{Haagerup79I}, is a bijection from the set of all normal semifinite weights on $\nalgebra$ onto the set of all normal semifinite weights $\psi$ on $\nalgebra\rtimes_\alpha\mathbb{R}$ satisfying
		$$\psi \circ \theta_t=\psi \quad \forall t\in \mathbb{R}. $$
	\end{lemma}
	
	Let $\tau$ be the unique trace described in Theorem \ref{Tcrossedproduct} $(ii)$. We can use this trace to construct a noncommutative measure on the set of $\tau$-measurable closed affiliated operators as presented in the very beginning of this section. In addition, for every normal semifinite weight $\phi$ on $\nalgebra$ we can obtain by the Radon-Nikodym derivative a operator $H_\phi$ affiliated with $\nalgebra\rtimes_\alpha \mathbb{R}$ such that $\tilde{\phi}=\tau_{H_\phi}$. \cite[Theorem 1.2]{haagerup79} states that $\{H_\phi \in \left(\nalgebra\rtimes_\alpha \mathbb{R}\right)_\eta \, | \, \phi \textrm{ is normal and semifinite} \}$ is exactly the set $\{H_\phi \in \left(\nalgebra\rtimes_\alpha \mathbb{R}\right)_\eta \, | \, \theta_t H_\phi=e^{-t}H_\phi \}$. This suggests the following definition:
	
	\begin{definition}\index{noncommutative $L_p$-space! Haagerup}
		Let $\nalgebra$ be a von Neumann algebra, $\phi_0$ a nornal faithful semifinite weight on $\nalgebra$ and $1\leq p<\infty$, we define	
		$$\begin{aligned}
		L_p(\nalgebra)&\doteq\left\{H\in\left(\nalgebra\rtimes_\alpha \mathbb{R}\right)_\tau \ \middle| \ \theta_t H=e^{-\frac{t}{p}}H, \ \forall t \in \mathbb{R} \right\} \textrm{ with } \|H\|_p=\tau(|H|^p)^\frac{1}{p},\\
		L_\infty(\nalgebra)&\doteq\left\{H\in\left(\nalgebra\rtimes_\alpha \mathbb{R}\right)_\tau \ \middle| \ \theta_t H=H, \ \forall t \in \mathbb{R} \right\} \textrm{ with } \|H\|_\infty=\|H\|.\\\end{aligned}$$
	\end{definition} 
	
	\cite[Theorem 1.14, Theorem 1.16 and Theorem 1.19]{haagerup79} states a H\"older inequality, a Minkowski inequality and the usual duality between spaces with conjugated indices (\ie an analogous of Theorem \ref{dualLp}) for these spaces. The proofs rely in the application of the same techniques we used previously.
	\fi
	
	\section{Perturbation of $p$-Continuous KMS States}
	\label{chapExtensionPerturb}
	\setcounter{section}{3}
	
	The idea of extending Araki's perturbation theory using noncommutative $L_p$-spaces was proposed by C. D. J\"akel and consists in a new approach to the problem. Now we start presenting the main results of this work. All that follows is entirely new.
	
	It is quite clear that one of the key properties used in \cite{Araki73} and \cite{sakai91} to prove the convergence of the Dyson's series, or in \cite{Araki73.2} to prove the convergence of the expansional, is that $\|A_1\ldots A_n\|\leq \|A_1\|\ldots \|A_n\|$, which is one of the axioms of Banach algebras. Unfortunately, this property does not hold in noncommutative $L_p$-spaces. In fact, in Banach algebras these are not even algebras under the induced multiplication. In particular, we have $\|Q^n\|\leq \|Q\|^n$, but no similar property holds in noncommutative $L_p$-spaces.

	\begin{proposition}
		\label{noconstant}
		Let $\nalgebra$ be a von Neumann algebra, $\tau$ be a normal faithful semifinite trace on $\nalgebra$, and $A\in L_1(\nalgebra,\tau)$. There exists $M>0$ such that $\tau\left(|A|^n\right)\leq M^n$ for all $n\in\mathbb{N}$,
		if, and only if, $A \in \nalgebra$.
	\end{proposition}
	\begin{proof}
		$(\Rightarrow)$ Let's prove the contrapositive.
		Suppose $A$ is unbounded and let \\${\displaystyle |A|=\int_0^\infty \lambda dE_\lambda^{|A|}}$ be the spectral decomposition of $|A|$.
		
		For every $K>M$, $E_{(K,\infty)}$ is non-null, so $\tau(E_{(K,\infty)})>0$. Then,
		$$\tau\left(|A|^n\right)=\int_0^\infty \lambda^n \tau\left(dE_\lambda^{|A|}\right) \geq \int_K^\infty \lambda^n \tau\left(dE_\lambda^{|A|}\right)\geq K^n \tau\left(E_{[K,\infty)}\right). $$
		
		Now, we already know that there exists $N\in \mathbb{N}$ large enough such that, for all $n\geq N$, $M^n< K^n \tau\left(E_{[K,\infty)}\right)$.
		
		$(\Leftarrow)$ The case $A=0$ is trivial. Suppose $A\neq0$ is bounded. Then
		$$\begin{aligned}
		\tau\left(|A|^n\right)&=\tau\left(|A|^{n-1}|A|\right)\\
		&\leq \left\||A|^{n-1}\right\|\tau\left(|A|\right)\\
		&=\left\|A\right\|^n\frac{\tau\left(|A|\right)}{\|A\|}\\
		&\leq \left(\|A\|\max\left\{1,\frac{\tau\left(|A|\right)}{\|A\|}\right\}\right)^n.
		\end{aligned}$$
		
	\end{proof}

	The next definition captures our intentions of having a convergent Dyson's series. In this definition, one subtle  difference is that the exponent cannot be passed out the trace, what is the $C^\ast$-condition for $p=\infty$. On the physical point of view, we do not want the high order terms in perturbation to affect our system too much, at least its integral.
	
	\begin{definition}
		\label{defex}
		Let $\nalgebra$ be a von Neumann algebra, $\tau$ be a normal faithful semifinite trace on $\nalgebra$, $1\leq p\leq\infty$ and $0<\lambda<\infty$. An operator $A \in L_p(\nalgebra,\tau)$ is said to be $(\tau,p,\lambda)$-exponentiable if 
		\begin{equation}
		\label{eq:defexponentiable}
		\sum_{n=1}^\infty\frac{\lambda^n\||A|^n\|_{p}}{n!}< \infty.
		\end{equation}
		Furthermore, an operator $A \in L_p(\nalgebra,\tau)$ is said to be $(\tau,p,\infty)$-exponentiable if 
		\begin{equation}
		\label{eq:defexponentiableinf}
		\sum_{n=1}^\infty\frac{\lambda^n\||A|^n\|_{p}}{n!}< \infty, \qquad \forall \lambda\in\mathbb{R}_+.
		\end{equation}
		We denote $$\ex^\tau_{p,\lambda}=\left\{A\in L_p\left(\nalgebra,\tau\right) \ \middle | \ A \textrm{ is } (\tau,p,\lambda)\textrm{-exponentiable } \right\}.$$
	\end{definition}

	Some properties can be seen directly from the definition. The first is that, if $\lambda\leq \lambda^\prime$, then $\ex^\tau_{p,\lambda}\subset\ex^\tau_{p,\lambda^\prime}$.
	Another very useful property that we will use to simplify our presentation is that $$\ex^\tau_{p,\lambda}=\lambda \ex^\tau_{p,1}=\left\{\lambda A\in L_p(\nalgebra) \ | A\in \ex^\tau_{p,1} \right\}.$$
	So, the only special case is $\ex^\tau_{p,\infty}$, for which we have $\displaystyle \ex^\tau_{p,\infty}=\bigcap_{\lambda\in\mathbb{R}_+}\ex^\tau_{p,\lambda}$.
	Hence, it is enough to study $\ex^\tau_{p,1}$ and $\ex^\tau_{p,\infty}$.
	
	\begin{notation}
		In order to simplify the notation, we will denote $\ex^\tau_{p,1}=\ex^\tau_p$ and call a $(\tau,p,1)$-exponentiable operator just a $(\tau,p)$-exponentiable operator.
	\end{notation}
	
	\begin{remark}
		Notice that equation \eqref{eq:defexponentiable} can be written in many forms for $1\leq p< \infty$
		$$\sum_{n=1}^\infty\frac{\||A|^n\|_{p}}{n!}=\sum_{n=1}^\infty\frac{\|A\|_{np}^n}{n!}=\sum_{n=1}^\infty\frac{\tau\left(|A|^{np}\right)^\frac{1}{p}}{n!}< \infty.$$
		
		We prefer equation \eqref{eq:defexponentiable} because it also includes the case $p=\infty$, for which
		$$\sum_{n=1}^{\infty} \frac{\||A|^n\|_{\infty}}{n!}\leq\sum_{n=1}^{\infty} \frac{\|A\|_{\infty}^n}{n!}=e^{\|A\|}-1<\infty.$$
		Hence, we have $\ex^\tau_\infty=\nalgebra$.
		
	\end{remark}
	
	In order so simplify calculations in our examples and constructions, we prove the following lemma.
	
	\begin{lemma}
		Let $\nalgebra$ be a von Neumann algebra and $\tau$ a normal faithful semifinite trace on $\nalgebra$. Then $A\in \ex^\tau_p$ if $$\sum_{n=1}^\infty\frac{1}{n!}\tau\left(|A|^{np}\right)< \infty.$$
	\end{lemma}
	\begin{proof}
		Define
		$$\begin{aligned}
		N_+=\left\{n\in\mathbb{N} \ \middle | \ \tau\left(\left|A\right|^{pn}\right)>1  \right\},\\
		N_-=\left\{n\in\mathbb{N} \ \middle | \ \tau\left(\left|A\right|^{pn}\right)\leq 1  \right\}.\\
		\end{aligned}$$
		
		It is clear that
		\begin{equation}
		\begin{aligned}
		\sum_{n=1}^{N}\frac{\|A\|_{np}^n}{n!}& =\sum_{n=1}^{N}\frac{1}{n!}\tau\left(\left|A\right|^{pn}\right)^\frac{1}{p}\\
		&=\sum_{\substack{n\in N_- \\ n\leq N}}\frac{1}{n!}\tau\left(\left|A\right|^{pn}\right)^\frac{1}{p}+\sum_{\substack{n\in N_+ \\ n\leq N}}\frac{1}{n!}\tau\left(\left|A\right|^{pn}\right)^\frac{1}{p}\\
		&\leq\sum_{n=1}^\infty\frac{1}{n!}+\sum_{n=1}^\infty\frac{1}{n!}\tau\left(\left|A\right|^{pn}\right).
		\end{aligned}
		\end{equation}
	\end{proof}
	
	The next step is to prove that the set we just defined is big, in some sense, in $L_p(\nalgebra,\tau)$.
	
	\begin{proposition}
		\label{ex_dense}
		$\ex^\tau_p$ and $\ex^\tau_\infty$ are $\|\cdot\|$-dense in $L_p(\nalgebra,\tau)$.
	\end{proposition}
	\begin{proof}
		It is enough to prove $\ex^\tau_p$ is dense $\|\cdot\|$-dense in $L_p(\nalgebra,\tau)$.  Let \mbox{$A \in L_p(\nalgebra,\tau)$} be a positive operator and let its spectral decomposition be \mbox{$\displaystyle A=\int_0^\infty \lambda dE_\lambda$}.
		Define $\displaystyle A_m=\int_0^m \lambda dE_\lambda$. Then, for all $n\in\mathbb{N}$,
		$$\begin{aligned}
		\tau\left(\left(A_m^p\right)^n\right)& =\int_0^m \lambda^{pn} \tau(dE_\lambda)\\
		&=\int_0^1 \lambda^{pn} \tau(dE_\lambda)+\int_1^m \lambda^{pn} \tau(dE_\lambda)\\
		&\leq\int_0^1 \lambda^{p} \tau(dE_\lambda)+m^{p(n-1)}\int_1^m \lambda^{p} \tau(dE_\lambda)\\
		& \leq m^{p(n-1)}\int_0^m \lambda^{p} \tau(dE_\lambda)\\
		&=m^{p(n-1)}\tau\left(|A_m|^p\right).
		\end{aligned}$$
		
		Hence $\left(A_m\right)_{m\in\mathbb{N}}$ is a sequence of $(\tau,p,\infty)$-exponentiable operators and
		$$\tau\left(|A-A_m|^p\right)=\int_m^\infty\lambda^p\tau\left(dE_\lambda\right)\xrightarrow{n\to \infty}0.$$
		
		For the general case, just remember the polarization identity implies every operator is a linear combination of four positive operators.
		
	\end{proof}
	
	Notice that Proposition \ref{ex_dense} shows that $\nalgebra \cap L_p(\nalgebra,\tau)\subset \ex^\tau_{p,\infty}$ and \\ ${\|A^n\|_p\leq \max\{1,\|A\|^{n-1}\|A\|_p\}}$ for $A\geq 0$. It is not difficult to see that the conclusion could also be obtained by Lemma \ref{normtraceinequality} and the well known result (see \cite[page 23]{terp81})
	$$\overline{\nalgebra \cap L_p(\nalgebra,\tau)}^{\|\cdot\|_p}=L_p(\nalgebra,\tau),$$
	which is in fact proved using an argument similar to what we have used above.
	
	This comment raises doubts about the possible ``triviality'' of $\ex^\tau_p$ or $\ex^\tau_{p,\infty}$, I mean, although we have already proved that these sets are big enough to be dense, the set we used to prove density consists of bounded operators. The next example will answer this question.

	\begin{example}
		\label{exfunction}
		Consider a function $f:\mathbb{R}\setminus\{0\}\to\mathbb{R}$ given by
		$$f(x)=\begin{cases}m & \textrm{if } \frac{1}{(m+1)!}\leq |x| < \frac{1}{m!}, \, m\in\mathbb{N}\\
		0 & \textrm{if } |x|\geq1.\\ \end{cases}$$
		
		This is a positive unbounded integrable function with compact support in $\mathbb{R}$ and, for each $\lambda>0$, 
		$$	\sum_{n=1}^\infty\frac{1}{n!}\int_{\mathbb{R}} \left(\lambda f(x)\right)^n dx =2\frac{\left(e^{e^\lambda}-1\right)\left(e^\lambda-1\right)}{e^\lambda}.$$
		Of course, we don't need the exact result and the reader can check the it is obvious that this sum would be less than $e^{e^\lambda}$.
		
		For a measurable set $K\in\mathbb{R}\setminus\{0\}$ such that $0$ is an accumulation point of $K$, the restriction of $f$ to $K$ is an example of an unbounded $\left(\int_{K}\cdot dx,1,\infty\right)$-exponentiable operator of $L_1\left(L_\infty(K),\int_{K}\cdot dx\right)$.
		
		It is obvious that any integrable function dominated by the previous one is also $\left(\int_{K}\cdot dx,1,\infty\right)$-exponentiable.
		
	\end{example}
	
	One could wonder if $\ex^\tau_p$ is a vector space or not. The following example shows the answer is negative. Another consequence of these \hyphenation{ex-am-ples}examples is that, for $\lambda<\lambda^\prime$, $\ex^\tau_{p,\infty}\subsetneq \ex^\tau_{p,\lambda}\subsetneq \ex^\tau_{p,\lambda^\prime}$. This is a very important and non trivial, since it means that $\left\{\ex^\tau_p\right\}_{p\in \overline{R}_+}$ or even  $\left\{\ex^\tau_{p,\lambda}\right\}_{p\in \overline{\mathbb{R}}_+, \lambda\in\overline{\mathbb{R}}_+}$ are, in general, non trivial gradations of $\nalgebra=\ex^\tau_{\infty,\lambda}=\ex^\tau_{\infty,\lambda}$ for every $\lambda\in\overline{\mathbb{R}}_+$.

	\begin{example}
		\label{exfunctionbadbehaved}
		Consider a function $f:\mathbb{R}\setminus\{0\}\to\mathbb{R}$ given by
		$$f(x)=\begin{cases}m & \textrm{if } (2e)^{-m-1}\leq |x| < (2e)^{-m}, \ m\in\mathbb{N},\\
		0 & \textrm{if } |x|\geq2e.\\ \end{cases}$$
		
		This is a positive unbounded integrable function with compact support in $\mathbb{R}$ and
		$$\sum_{n=1}^\infty\frac{1}{n!}\int_{\mathbb{R}} f(x)^n dx =\frac{2(e-1)}{2}.$$
		
		Hence, again we have that for any measurable set $K\in\mathbb{R}\setminus\{0\}$ the restriction of $f$ \mbox{to $K$} is an example of a $\left(\int_{K}\cdot dx\right)$-exponentiable operator for $L_1\left(L_\infty(K),\int_{K}\cdot dx\right)$, but it does not hold for $2f$. In fact,
		
		$$\sum_{n=1}^N\frac{1}{n!}\int_{\mathbb{R}} f(x)^{n} dx=\frac{4e}{2e-1}\sum_{n=1}^N\frac{1}{n!}\sum_{m=1}^{\infty} (2m)^{n} (2e)^{-m}\\$$
		is a divergent series.
		
	\end{example}
	
	Although we have presented examples just for $p=1$, it is enough to take the $p$-th root of $f$ to obtain examples for any $p>1$.
	
	\begin{example}
		In order to construct an example in a noncommutative von Neumann algebra it is sufficient that there exists a monotonic decreasing sequence of projections $(P_n)_{n\in\mathbb{N}}\in\nalgebra_p$ such that $P_n \xrightarrow[]{\|\cdot\|_1} 0$, which is true if there exists any $\tau$-measurable unbounded operator.
		
		In fact, fix $1\leq p$. If there exists such a sequence, we can suppose without loss of generality, by taking a subsequence if necessary, that $\tau(P_n)\leq \frac{1}{(e^n-1)2^n}$. Define the positive unbounded $\tau$-measurable operator $$A=\sum_{n=1}^\infty n^\frac{1}{p}(P_n-P_{n+1}).$$
		
		It follows from the definition that
		$$\begin{aligned}
		\sum_{m=1}^\infty \frac{1}{m!}\tau(|A|^{pm})&=\sum_{m=1}^\infty \frac{1}{m!}\sum_{n=1}^\infty n^m\tau(P_n-P_{n+1})\\
		&\leq \sum_{m=1}^\infty \sum_{n=1}^\infty \frac{n^m}{m!} \frac{1}{(e^n-1)2^n}\\
		&= \sum_{n=1}^\infty \frac{1}{2^n}=1.
		\end{aligned}$$
		Thus $A\in\ex^\tau_p$.
	\end{example}
	
	It is not difficult to see, with help of the spectral decomposition, that if an operator is in $L_p(\nalgebra,\tau)\cap L_q(\nalgebra,\tau)$ with $1\leq p < q < \infty$, then it is in $L_r(\nalgebra,\tau)$ for every $p\leq r \leq q$. More than that, it follows by analyticity and the Three-Line Theorem, a special case of the Riesz-Thorin Theorem, that:
	\begin{equation}
	\label{eq:riesz-thorin}
	\begin{aligned}
	&\|A\|_r \leq \|A\|_p^{\frac{p}{(q-p)}\left(\frac{q}{r}-1\right)}\|A\|_q^{{\frac{q}{(q-p)}\left(1-\frac{p}{r}\right)}},& \quad \textrm{ if } q<\infty;\\
	&\|A\|_r \leq \|A\|_p^\frac{p}{r}\|A\|_\infty^{1-\frac{p}{r}},& \quad \textrm{ if } q=\infty.
	\end{aligned}
	\end{equation}
	
	An analogous property holds for $(\tau,p)$-exponentiable operators:
	
	\begin{proposition}
		Let $1\leq p < q \leq \infty$, then $\ex^\tau_p\cap\ex^\tau_q\subset \ex^\tau_r$ for every $p\leq r\leq q$.
	\end{proposition}
	\begin{proof}
		Let $A\in \ex^\tau_p\cap\ex^\tau_q$, in particular $|A|^n\in L_p(\nalgebra,\tau)\cap L_q(\nalgebra,\tau)$ for all $n\in\mathbb{N}$.
		
		Using equation \eqref{eq:riesz-thorin} we get that, for all $p\leq r\leq q$,
		$$\||A|^n\|_r\leq \max\bigl\{\||A|^n\|_p,\||A|^n\|_q\bigr\}\leq \||A|^n\|_p+\||A|^n\|_q ,$$
		$$\begin{aligned}
		\sum_{n=1}^N \frac{\||A|^n\|_r}{n!}&=\sum_{n=1}^N \frac{\||A|^n\|_p+\||A|^n\|_q}{n!}\\
		&=\sum_{n=1}^\infty \frac{\||A|^n\|_p}{n!}+\sum_{n=1}^\infty \frac{\||A|^n\|_q}{n!}\\
		&<\infty.
		\end{aligned}$$	
	\end{proof}
	
	Although $\ex^\tau_p$, in general, are not vector spaces, they still have a very convenient geometric structure for perturbations.
	
	\begin{proposition}
		\label{exconvex}
		\begin{enumerate}[(i)]
			\item $\ex^\tau_p$ is a balanced and convex set;
			
			\item for every $A\in\ex^\tau_p$ and $B\in \nalgebra$ with $\|B\|\leq 1$, $BA\in \ex^\tau_p$;
			
			\item if $1\leq p, \, q,\, r\leq \infty$ are such that $\frac{1}{p}+\frac{1}{q}=\frac{1}{r}$ and $A,B \in \nalgebra_\tau$,
			$$\sum_{n=1}^\infty\frac{\tau(|A|^{np})}{n!} \quad , \quad \sum_{n=1}^\infty\frac{\tau(|B|^{nq})}{n!}<\infty \ \Rightarrow \ \sum_{n=1}^\infty\frac{\tau(|AB|^{nr})}{n!}<\infty;$$
			
			\item $\ex^\tau_{p,\infty}$ is a subspace of $L_p(\nalgebra)$.
		\end{enumerate}
	\end{proposition}
	\begin{proof}
		$(i)$ 	It is obvious that, for $A\in\ex^\tau_p$ and $|\lambda|\leq 1$ we have
		$$\begin{aligned}
		\sum_{n=1}^\infty\frac{\|\lambda A\|_{np}^n}{n!} &=\sum_{n=1}^\infty\frac{|\lambda|^n \|A\|_{np}^n}{n!}\leq\sum_{n=1}^\infty\frac{\|A\|_{np}^n}{n!}.\\
		\end{aligned}$$	
		
		Let $A,B\in \ex^\tau_p$ and let $0<\lambda<1$. Then, thanks to convexity of $x\mapsto x^n$ for each $n\in\mathbb{N}$,
		$$\begin{aligned}
		\sum_{n=1}^\infty\frac{\|\lambda A +(1-\lambda)B\|_{np}^n}{n!}&\leq\sum_{n=1}^\infty\frac{1}{n!}\left(\lambda \|A\|_{np}+(1-\lambda)\|B\|_{np}\right)^n\\
		&\leq\sum_{n=1}^\infty\frac{1}{n!}\left(\lambda\|A\|_{np}^{n}+(1-\lambda)\|B\|_{np}^{n}\right)\\
		&\leq\sum_{n=1}^\infty\frac{1}{n!}\|A\|_{np}^{n}+\sum_{n=1}^\infty\frac{1}{n!}\|B\|_{np}^{n};\\
		\end{aligned}$$
		
		$(ii)$ It follows trivially from $(i)$ in  Theorem \ref{minkowski};
		
		$(iii)$ It follows from Corollary \ref{g2holder} that 
		$$\begin{aligned}
		\sum_{n=1}^N\frac{\tau\left(|AB|^{nr}\right)}{n!}&\leq \sum_{n=1}^N\frac{1}{n!}\tau\left(|A|^{np}\right)^{\frac{r}{p}}\tau\left(|B|^{nq}\right)^{\frac{r}{q}}\\
		&=\sum_{n=1}^N\left(\frac{\tau\left(|A|^{np}\right)}{n!}\right)^\frac{r}{p}\left(\frac{\tau\left(|B|^{nq}\right)}{n!}\right)^\frac{r}{q}\\
		&\leq\left(\sum_{n=1}^N\frac{\tau\left(|A|^{np}\right)}{n!}\right)^\frac{r}{p}\left(\sum_{n=1}^N\frac{\tau\left(|B|^{nq}\right)}{n!}\right)^\frac{r}{q}.\\
		\end{aligned}$$
		
		$(iv)$ Notice that $A\in\ex^\tau_{p,\infty}$ if, and only if, $\lambda A\in\ex^\tau_p$ for every $\lambda\in\mathbb{R}$. If $\alpha=\beta=0$ the result is obvious, otherwise, it follows from item $(i)$ that, if $\alpha, \, \beta \in \mathbb{C}$ and $A,B\in\ex^\tau_{p,\infty}$,$$\begin{aligned}\alpha A+\beta B&=\left(|\alpha|+|\beta| \right)\left(\frac{|\alpha|}{|\alpha|+|\beta|}\left(\frac{\alpha}{|\alpha|}A\right)+\frac{|\beta|}{|\alpha|+|\beta|}\left(\frac{\beta}{|\beta|}B\right)\right)\\ &\in \left(|\alpha|+|\beta| \right)\ex^\tau_p\subset \ex^\tau_{p,\infty}.\end{aligned}$$
	\end{proof}
	
	The following lemma justifies the choice of the name ``exponentiable'' for such operators.
	
	\begin{lemma}
		\label{ExpansionalinLp}
		For each $A\in\ex^\tau_{p,\lambda}$ and $B\eta \nalgebra$ self-adjoint, define $A(t)=B^{it}AB^{-it}$. Then, for $0\leq t < \lambda$,
		$$\mathbbm{1}-Exp_r\left(\int_0^t;A(s) ds\right) \quad \textrm{ and } \quad \mathbbm{1}- Exp_l\left(\int_0^t;A(s) ds\right)\in L_p\left(\nalgebra,\tau\right).$$
	\end{lemma}
	\begin{proof}
		
		Since $A\in \nalgebra_\tau \Rightarrow A(t)\in \nalgebra_\tau$, Proposition \ref{measurablealgebra} implies that each term in the definition of these operators is in $\nalgebra_\tau$ except for the identity. 
		
		In addition, using the Theorem \ref{gholder} for $p_i=n$, $i=1, \ldots, n$, we have, for every $N>M$, that
		
		\begin{equation}
		\label{eq:calculation1}
		\begin{aligned}
		\Bigg\|\sum_{n=N}^{M}\int_{0}^{t} dt_1\ldots \int_{0}^{t_{n-1}}dt_{n}&A(t_n)\ldots A(t_1)\Bigg\|_p\\	&\leq\sum_{n=N}^{M}\left\|{\int_{0}^{t} dt_1\ldots \int_{0}^{t_{n-1}}dt_{n}A(t_n)\ldots A(t_1)}\right\|_p\\	
		&\leq\sum_{n=N}^{M}\tau\left(\left|{\int_{0}^{t} dt_1\ldots \int_{0}^{t_{n-1}}dt_{n}A(t_n)\ldots A(t_1)}\right|^p\right)^\frac{1}{p}\\
		&\leq\sum_{n=N}^{M}\int_{0}^{t} dt_1\ldots \int_{0}^{t_{n-1}}dt_{n}\tau\left(\left|A(t_n)\ldots A(t_1)\right|^p\right)^\frac{1}{p}\\
		&=\sum_{n=N}^{M}\int_{0}^{t} dt_1\ldots \int_{0}^{t_{n-1}}dt_{n}\tau\left(\left|A\right|^{pn}\right)^\frac{1}{p}\\
		&=\sum_{n=N}^{M}\frac{t^n}{n!}\tau\left(\left|A\right|^{pn}\right)^\frac{1}{p},\\
		\end{aligned}
		\end{equation}
		which shows simultaneously that each term is in $L_p(\nalgebra,\tau)$ and the partial sum is a $\|\cdot\|_p$-Cauchy sequence. The thesis follows by completeness.
		
	\end{proof}
	
	To clarify the next definition, remember that $\nalgebra\cap L_p(\nalgebra,\tau)$ (or even $\nalgebra\cap L_1(\nalgebra,\tau)$) is $\|\cdot\|_p$-dense in $L_p(\nalgebra,\tau)$.
	
	\begin{definition}
		Let $\nalgebra$ be a von Neumann algebra and $\tau$ a faithful normal semifinite trace on $\nalgebra$. We say that a state $\phi$ on $\nalgebra$ is $\|\cdot\|_p$-continuous if it is continuous on $\left(\nalgebra\cap L_p(\nalgebra,\tau),\|\cdot\|_p\right)$.
		
		Of course, such a weight can be continuously extended to $\left(L_p(\nalgebra,\tau),\|\cdot\|_p\right)$ in a unique way.
	\end{definition}
	
	Hitherto, we have defined a set, namely $\ex^\tau_p$, that we assert is the right set to take our perturbation, but the reader should be warned after so many comments about the duality relations between $L_p$-spaces that we will demand some extra ``dual'' property on the original state. This motivates our next definition.
	
	\pagebreak
	
	\begin{proposition}
		\label{defpcont}
		Let $\nalgebra$ be a von Neumann algebra and $\tau$ a faithful normal semifinite trace on $\nalgebra$ and $1\leq p, \, q\leq \infty$ such that $\frac{1}{p}+\frac{1}{q}=1$. A state $\phi$ on $\nalgebra$ is $\|\cdot\|_p$-continuous if, and only if, there exists $H\in L_q(\nalgebra,\tau)$, $H\eta \mathfrak{M}_{\tau^\phi}$, such that 
		$$\phi(A)=\tau_H(A) \quad \forall A\in\nalgebra,$$ in the sense of Proposition \ref{unboundedderivativeweight}.
	\end{proposition}
	\begin{proof}
		As mentioned in Definition \ref{defpcont}, we can continuously extended $\phi$ to $L_p(\nalgebra,\tau)$. By the dual relation between $L_p(\nalgebra,\tau)$ and $L_q(\nalgebra,\tau)$ stated in Theorem \ref{dualLp}, there exists $H\in L_q(\nalgebra,\tau)$ such that $\phi(A)=\tau(HA)$ for all $A\in L_p(\nalgebra,\tau)$. Such $H$ must be affiliated with $\mathfrak{M}_{\tau^\phi}$.
		
		In particular, $\phi(A)=\tau(HA)$ for all $A\in \nalgebra\cap L_p(\nalgebra,\tau)$, but $\nalgebra\cap L_p(\nalgebra,\tau)$ is \mbox{WOT-dense} in $\nalgebra$, since the trace is semifinite.
		
		The cases $p=1,\infty$, are analogous and the other part of the equivalence is trivial.
		
	\end{proof}
	
	The next two theorems can be seen as the key to guarantee Dyson's series is convergent. It is time to stress how important Araki's multiple-time KMS condition is for the theory. Here it is used with the same purposes of the original Araki's article \cite{Araki73}. Mentioning an interesting connection, this property is also used in the Araki's noncommutative $L_p$-spaces, what makes us believe there is a natural way to extend this result.
	
	\begin{notation} \begin{enumerate}[(i)]
		
		\item Let $\Omega\subset \mathbb{R}^n$ an convex domain (\ie \, an open convex set). We define the tube over $\Omega$ by $$T(\Omega)=\left\{z\in \mathbb{C}^n \ \middle| \ \Im{z}\in \Omega \right\}.$$
		
		The following convex domain will play a relevant role to our purposes	$$S^n_{\alpha}\doteq\left\{(t_1,\ldots,t_n)\in \mathbb{R}^n \ \middle | \ t_i<0, \ 1\leq i\leq n, \textrm{ and } -\alpha<\sum_{i=1}^n t_i<0\right\}.$$
		
		\item $\nanalytic$ will denote the set of analytic elements for the (one-parameter) modular automorphism group.
	\end{enumerate}
	\end{notation}
	
	\begin{theorem}
		\label{TR0}
		Let $\nalgebra$ be a von Neumann algebra, $\phi$ a faithful state in $\nalgebra$ and $n\in\mathbb{N}$. Let also $\left(\hilbert_\phi,\Phi, \pi_\phi\right)$ be the GNS representation throughout $\phi$, $\tau$ a normal faithful semifinite trace on $B(\hilbert_\phi)$, $Q_i, J_\phi Q_i J_\phi\in L_{2mq}\left(B(\hilbert_\phi),\tau\right)$ such that $\|J_\phi Q_i J_\phi\|_{2mq}=\|Q_i\|_{2mq}$ for all $1\leq i,m\leq n$ and suppose $\phi=\tau_H$ is $\|\cdot\|_p$-continuous. Then, if $\Phi \in\Dom{Q_1}$ and $\Delta_\Phi^{\iu z_{j-1}}Q_{j-1}\ldots \Delta_\Phi^{\iu z_1}Q_1\Phi \in \Dom{Q_j}$ for every $-\frac{1}{2}\leq \Im{z_j} \leq 0$ and for every $2\leq j\leq n$,
		
		$$Q_n\Delta_\Phi^{\iu z_{n-1}}Q_{n-1}\ldots \Delta_\Phi^{\iu z_1}Q_1\Phi \in \Dom{\Delta_\Phi^{\iu z}} \textrm{ for } -\frac{1}{2}\leq\Im{z}\leq0 \textrm{ and }$$
		$$A^n(z_1,\ldots, z_n)\Phi \doteq \Delta_\Phi^{\iu z_n} Q_n \Delta_\Phi^{\iu z_{n-1}}Q_{n-1}\ldots \Delta_\Phi^{\iu z_1}Q_1\Phi$$ is analytic on $T\left(S^n_\frac{1}{2}\right)$
		and bounded on its closure by
		$$ \left\|A^n(z_1,\ldots, z_n)\Phi\right\|\leq \|H\|_p^\frac{1}{2} \prod_{j=1}^{n}\|Q_i\|_{2nq}.$$
	\end{theorem}
	
	\begin{proof}
		Let's proceed by induction on $n$.
		
		For $n=1$, let $Q_1=U|Q_1|$ be the polar decomposition of $Q_1$ and $\displaystyle|Q_1|=\int_0^\infty \lambda dE^{|Q_1|}_\lambda$ the spectral decomposition of $|Q_1|$. Since $\Phi\in\Dom{Q_1}$, $\displaystyle Q_1\Phi=U\lim_{k\to \infty} Q_{1,k}\Phi$, where $\displaystyle Q_{1,k}=\int_0^k \lambda dE^{|Q_1|}_\lambda$.
		Define the following \hyphenation{func-tion-als}functionals on $\nanalytic\Phi$:
		$$\begin{aligned}
		f_k^z(A\Phi) &\doteq\ip{UQ_{1,k}\Phi}{\Delta_\Phi^{-\iu \bar{z}}A\Phi}_\phi,\\
		f^z(A\Phi)&\doteq\lim_{k\to\infty}f_k^z(A\Phi)=\ip{Q_1\Phi}{\Delta_\Phi^{-\iu \bar{z}}A\Phi}_\phi.\\	
		\end{aligned}$$
		
		Of course, for fixed $A\Phi$, $\bar{f_k}(z)=\overline{f_k^z(A\Phi)}$ is entire analytic and
		\begin{equation}
		\label{calculationX17}
		\begin{aligned}
		\left|\bar{f}(t)\right|&=\lim_{k\to\infty}\left|\bar{f}_k(t)\right|\\
		&=\lim_{k\to\infty}\left|\ip{\Delta_\Phi^{-\iu t}A\Phi}{UQ_{1,k}\Phi}_\phi\right|\\
		&\leq \left\|\Delta_\Phi^{-\iu t}A\Phi\right\| \lim_{k\to\infty} \|UQ_{1,k}\Phi\|\\
		& \leq\left\|A\Phi\right\|\lim_{k\to\infty}\phi\left(Q_{1,k} U^\ast U Q_{1,k} \right)^\frac{1}{2}\\
		&\leq \left\|A\Phi\right\|\lim_{k\to\infty}\tau\left(H^\frac{1}{2}Q_{1,k} U^\ast U Q_{1,k} H^\frac{1}{2}\right)^\frac{1}{2}\\
		& \leq\left\|A\Phi\right\|\|H\|_p^\frac{1}{2}\||Q_1|^2\|_q^\frac{1}{2}\\
		&\leq\left\|A\Phi\right\|\|H\|_p^\frac{1}{2}\|Q_1\|_{2q}, \quad \forall t\in \mathbb{R}.\\
		\end{aligned}\end{equation}
		Moreover,	
		\begin{equation}
		\label{calculationX17C}
		\begin{aligned}
		\left|\bar{f}\left(t+\frac{1}{2}\iu\right)\right|&=\lim_{k\to\infty}\left|\bar{f}_k\left(t+\frac{1}{2}\iu\right)\right|\\
		&=\lim_{k\to\infty}\ip{\Delta_\Phi^{\frac{1}{2}}\Delta_\Phi^{-\iu t}A\Phi}{UQ_{1,k}\Phi}_\phi\\
		&\leq \left\|\Delta_\Phi^{-\iu t}A\Phi\right\| \lim_{k\to\infty} \|J_\Phi Q_{1,k} U^\ast\Phi\|\\
		&= \left\|A\Phi\right\| \lim_{k\to\infty} \| Q_{1,k} U^\ast\Phi\|\\
		&\leq \left\|A\Phi\right\|\lim_{k\to\infty}\phi\left(U Q_{1,k} Q_{1,k} U^\ast \right)^\frac{1}{2}\\
		&\leq \left\|A\Phi\right\|\lim_{k\to\infty}\tau\left(H^\frac{1}{2}U Q_{1,k}^2 U^\ast H^\frac{1}{2}\right)^\frac{1}{2}\\
		&\leq \left\|A\Phi\right\|\|H\|_p^\frac{1}{2}\||Q_1^\ast|^2\|_q^\frac{1}{2}\\
		&\leq \left\|A\Phi\right\|\|H\|_p^\frac{1}{2}\||Q_1|^2\|_q^\frac{1}{2}\\
		&\leq\left\|A\Phi\right\|\|H\|_p^\frac{1}{2}\|Q_1\|_{2q}, \quad \forall t\in \mathbb{R};\\	
		\end{aligned}
		\end{equation}
		which proves that the functional concerned is bounded for $-\frac{1}{2}\leq\Im{z}\leq0 $ due to the Maximum Modulus Principle. This bound also proves that if $\bar{f}_k\to \bar{f}$ uniformly for $-\frac{1}{2}\leq\Im{z}\leq0$, then $\bar{f}$ is analytic for $-\frac{1}{2}<\Im{z}<0$ and bounded for $-\frac{1}{2}\leq\Im{z}\leq0 $.
		
		Using first the Hahn-Banach Theorem to obtain an extension (also denoted by $f_z$) to the whole Hilbert space in such a way that $\|f_z\|\leq \|H\|_p^\frac{1}{2}\|Q_1\|_{2q}$, we know by the Riesz Representation Theorem, that there exists a $\Omega(z)\in \hilbert_\phi$ such that $f_z(\cdot)=\ip{\Omega(z)}{\cdot}_\phi$. Since $\nanalytic\Phi$ is dense, $\Omega(z)$ is unique. 
		
		So far we have that $Q_1\Phi\in \Dom{\left(\Delta_\Phi^{-\iu \bar{z}}\right)^\ast}=\Dom{\Delta_\Phi^{\iu z}}$ and $\bar{f}(z)=\ip{A\Phi}{\Delta_\Phi^{\iu z}Q_1\Phi}_\phi$ is analytic on $\left \{z\in\mathbb{C} \ \middle| \ -\frac{1}{2}< \Im{z}< 0\right\}$ and continuous on its closure, for every ${A\in \nanalytic}$.
		
		Since $\overline{\nanalytic\Phi}^{\|\cdot\|}=\overline{\nalgebra\Phi}^{\|\cdot\|}=\hilbert_\phi$, the vector-valued function $A(z)\Phi\doteq\Delta_\Phi^{\iu z}Q_1\Phi$ is weak analytic, hence, strong analytic on $\left \{z\in\mathbb{C} \ \middle| \ -\frac{1}{2}< \Im{z}< 0\right\}$ and
		$$\|A(z)\Phi\|\leq \|H\|_p^\frac{1}{2}\|Q_1\|_{2q} \quad \forall z\in \left \{z\in\mathbb{C} \ \middle| \ -\frac{1}{2}\leq \Im{z}\leq 0\right\}.$$
		
		Suppose now the hypothesis hold for $n\in \mathbb{N}$. We will use the same ideas: let $Q_{n+1}=U|Q_{n+1}|$ be the polar decomposition of $Q_{n+1}$ and $\displaystyle|Q_{n+1}|=\int_0^\infty \lambda dE^{|Q_{n+1}|}_\lambda$ the spectral decomposition of $|Q_{n+1}|$. Since $\Phi\in\Dom{Q_{n+1}}$, $\displaystyle Q_{n+1}\Phi=U\lim_{k\to \infty} Q_{n+1,k}\Phi$.
		
		$$
		f^{(z_1,\ldots, z_{n+1})}(A\Phi) =\ip{Q_{n+1}\Delta_\Phi^{\iu z_n}Q_n\ldots\Delta_\Phi^{\iu z_1}Q_1\Phi}{\Delta_\Phi^{-\iu \bar{z}_{n+1}}A\Phi}_\phi.$$
		
		Since $\bar{f}_k(z_1,\ldots,z_{n+1})\doteq\overline{f_k^{(z_1,\ldots, z_{n+1})}(A\Phi)}$ is an analytic function, it attains its maximum at an extremal point of $S^{n+1}_\frac{1}{2}$ (see \cite{Araki73} Corollary 2.2). Denoting $z_j=x_j+\iu y_j$, $x_j,\, y_j \in \mathbb{R}$ for all $1\leq j\leq n+1$, and repeating the calculations in equations \eqref{calculationX17} and \eqref{calculationX17C}, first for the extremal points with $\Im{z_j}=0$ for all $1\leq j\leq n+1$, we get
		
		$(i)$ if $\Im{z_i}=0, \ 1\leq i\leq n$,
		$$\begin{aligned}
		\left|\bar{f}_k(z_1,\ldots, z_{n+1})\right|	&=\left|\ip{Q_{n+1}\Delta_\Phi^{\iu x_n}Q_n\ldots\Delta_\Phi^{\iu x_1}Q_1\Phi}{\Delta_\Phi^{-\iu x_{n+1}}A\Phi}_\phi\right|\\
		&\leq\left\|Q_{n+1}\tau^\phi_{x_n}(Q_n)\ldots\tau^\phi_{x_n+\cdots+x_1}(Q_1)\Phi\right\| \left\|\Delta_\Phi^{-\iu x_{n+1}}A\Phi\right\|\\
		&\leq\tau\left(\left| Q_{n+1}\tau^\phi_{x_n}(Q_n)\ldots\tau^\phi_{x_n+\cdots+x_1}(Q_1)H^\frac{1}{2}\right|^2\right)^\frac{1}{2} \|A\Phi\|\\
		&\leq \left\|H^\frac{1}{2}\right\|_{2p}\ \prod_{i=1}^{n+1}\|Q_i\|_{2nq}.
		\end{aligned}$$
		
		$(ii)$ if $\Im{z_i}=0, \ 1\leq i\leq n$, $i\neq k$ and $\Im{z_k}=-\frac{1}{2}$, where $x_i=\Re{z_i}$
		
		$$\begin{aligned}
		&\left|\bar{f}_k(z_1,\ldots, z_{n+1})\right|\\
		&=\left|\ip{Q_{n+1}\Delta_\Phi^{\iu x_n}Q_n\ldots\Delta_\Phi^{\iu x_{k-1}}Q_{k-1}\Delta_\Phi^{\frac{1}{2}}Q_k\Delta_\Phi^{\iu x_1}Q_1\Phi}{\Delta_\Phi^{-\iu x_{n+1}}A\Phi}_\phi\right|\\
		&\leq \|A\Phi\| \left\|H\right\|_p^\frac{1}{2} \prod_{i=1}^{k-1}\left\|\tau^\phi_{x_n+\cdots+x_i}(Q_i)\right\|_{2nq} \ \prod_{i=k}^{n+1}\left\|J_\phi\tau^\phi_{x_n+\cdots+x_i}(Q_i)J_\phi\right\|_{2nq}\\
		&= \|A\Phi\| \|H\|_p^\frac{1}{2} \prod_{i=1}^{n+1}\left\|Q_i\right\|_{2nq}\\
		\end{aligned}$$
	\end{proof}
	
	The previous result depends a lot on the possibility of ``extending'' the trace, that is originally defined only in the algebra, to the algebra generated by $\nalgebra\cup\nalgebra^\prime$. One may try to define
	$$\tau(J_\Phi |A| J_\Phi B)=\tau(|A|)\tau(|B|),$$
	but it immediately fails, in general, since the application of this formula with either $A=\mathbbm{1}$ or $B=\mathbbm{1}$, due to $\tau(\mathbbm{1})=\infty$.
	
	In order to relax the condition on the possibility of having a trace in all the $GNS$-represented algebra, we have to demand more regularity on the perturbation. The next theorem shows almost the same result as the previous one, with a little more restricted perturbation.
	
	\begin{theorem}
		\label{TR1}
		Let $\nalgebra\subset B(H)$ be a von Neumann algebra, $\tau$ a normal faithful semifinite trace on $\nalgebra$, $\phi(\cdot)=\ip{\Phi}{\cdot\Phi}$ a state on $\nalgebra$ and $n\in\mathbb{N}$. Let also $n\in\mathbb{N}$, $p, \, q\geq1$ with $\frac{1}{p}+\frac{1}{q}=1$, \mbox{$Q_i\in L_{4mq}\left(\nalgebra,\tau\right)$} for all $1\leq i,m\leq n$ and suppose $\phi=\tau_H$ is $\|\cdot\|_p$-continuous.
		
		Then, if $\Phi \in\Dom{Q_1}$ and $\Delta_\Phi^{\iu z_{j-1}}Q_{j-1}\ldots \Delta_\Phi^{\iu z_1}Q_1\Phi \in \Dom{Q_j}$ for every $-\frac{1}{2}\leq \Im{z_j} \leq 0$ and for every $2\leq j\leq n$,
		$Q_n\Delta_\Phi^{\iu z_n}Q_{n-1}\ldots \Delta_\Phi^{\iu z_1}Q_1\Phi \in \Dom{\Delta_\Phi^{\iu z}}$ for $-\frac{1}{2}\leq\Im{z}\leq 0$ and $$A^n(z_1,\ldots, z_n)\Phi \doteq Q\Delta_\Phi^{\iu z_n}Q_n\ldots \Delta_\Phi^{\iu z_1}Q_1\Phi$$ is analytic on $T\left(S^n_\frac{1}{2}\right)$
		and bounded on its closure by
		$$ \left\|A^n(z_1,\ldots, z_n)\Phi\right\| \leq \|H\|_p^\frac{1}{2}\max_{0\leq l\leq n-1}\left\{\underbrace{\left(\prod_{j=1}^{l}\|Q_j\|_{4lq}\right)}_{=1 \textrm{ if } l=0}\left(\prod_{j=l+1}^{n}\|Q_j\|_{4(n-l)q}\right)\right\}\\.$$
	\end{theorem}
	\begin{proof}
		Let's proceed by induction on $n$.
		
		For $n=1$, let $Q_1=U|Q_1|$ be the polar \hspace{.1pt} decomposition of $Q_1$ and \mbox{$\displaystyle|Q_1|=\int_0^\infty \lambda dE^{|Q_1|}_\lambda$} the spectral decomposition of $|Q_1|$. Since $\Phi\in\Dom{Q_1}$, $\displaystyle Q_1\Phi=U\lim_{k\to \infty} Q_{1,k}\Phi$, where $\displaystyle Q_{1,k}=\int_0^k \lambda dE^{|Q_1|}_\lambda$.
		Define the following functionals on $\nanalytic\Phi$
		$$\begin{aligned}
		f_k^z(A\Phi) &\doteq\ip{UQ_{1,k}\Phi}{\Delta_\Phi^{-\iu \bar{z}}A\Phi}_\phi,\\
		f^z(A\Phi)&\doteq\lim_{k\to\infty}f_k^z(A\Phi)=\ip{Q_1\Phi}{\Delta_\Phi^{-\iu \bar{z}}A\Phi}_\phi.\\	
		\end{aligned}$$
		
		Of course, for fixed $A\Phi$, $\bar{f_k}(z)=\overline{f_k^z(A\Phi)}$ is entire analytic and, in the two lines of extremal points, we have
		\begin{equation}
		\label{calculationX17.1}
		\begin{aligned}
		\left|\bar{f}(t)\right|&=\lim_{k\to\infty}\left|\bar{f}_k(t)\right|\\
		&\leq\left\|A\Phi\right\|\|H\|_p^\frac{1}{2}\|Q_1\|_{2q} \quad \forall t\in \mathbb{R};\\	
		\left|\bar{f}\left(t+\frac{1}{2}\iu\right)\right|&=\lim_{k\to\infty}\left|\bar{f}_k\left(t+\frac{1}{2}\iu\right)\right|\\
		&=\lim_{k\to\infty}\ip{\Delta_\Phi^{\frac{1}{2}}\Delta_\Phi^{-\iu t}A\Phi}{UQ_{1,k}\Phi}_\phi\\
		&\leq\left\|A\Phi\right\|\|H\|_p^\frac{1}{2}\|Q_1\|_{2q} \quad \forall t\in \mathbb{R};\\	
		\end{aligned}
		\end{equation}
		which proves that the functional concerned is bounded for $-\frac{1}{2}\leq\Im{z}\leq0 $ due to the Maximum Modulus Principle. This bound also proves that if $\bar{f}_k\to \bar{f}$ uniformly for $-\frac{1}{2}\leq\Im{z}\leq0$, then $\bar{f}$ is analytic for $-\frac{1}{2}<\Im{z}<0$ and bounded for $-\frac{1}{2}\leq\Im{z}\leq0 $.
		
		Using first Hahn-Banach Theorem to obtain an extension, also denoted by $f_z$, to the whole Hilbert space in such a way that $\|f_z\|\leq \|H\|_p^\frac{1}{2}\|Q_1\|_{2q}$, we know by the Riesz's Representation Theorem, that there exists a $\Omega(z)\in \hilbert_\phi$ such that $f_z(\cdot)=\ip{\Omega(z)}{\cdot}_\phi$. Since $\nanalytic\Phi$ is dense, $\Omega(z)$ is unique. 
		
		So far we have that $Q_1\Phi\in \Dom{\left(\Delta_\Phi^{-\iu \bar{z}}\right)^\ast}=\Dom{\Delta_\Phi^{\iu z}}$ and $\bar{f}(z)=\ip{A\Phi}{\Delta_\Phi^{\iu z}Q_1\Phi}_\phi$ is analytic on $\left \{z\in\mathbb{C} \ \middle| \ -\frac{1}{2}< \Im{z}< 0\right\}$ and continuous on its closure, for every ${A\in \nanalytic}$.
		
		Since $\overline{\nanalytic\Phi}^{\|\cdot\|}=\overline{\nalgebra\Phi}^{\|\cdot\|}=\hilbert_\phi$, the vector-valued function $A(z)\Phi\doteq\Delta_\Phi^{\iu z}Q_1\Phi$ is weakly analytic, hence, strongly analytic on $\left \{z\in\mathbb{C} \ \middle| \ -\frac{1}{2}< \Im{z}< 0\right\}$ and
		$$\|A(z)\Phi\|\leq \|H\|_p^\frac{1}{2}\|Q_1\|_{2q} \quad \forall z\in \left \{z\in\mathbb{C} \ \middle| \ -\frac{1}{2}\leq \Im{z}\leq 0\right\}.$$
		
		Suppose now the hypothesis hold for $n\in \mathbb{N}$. We will use the same ideas: we can define the sequence $\displaystyle  Q^{k_{i}}_i=U_i \int_0^{k_i} \lambda dE^{|Q_{i}|}_\lambda$, where $Q_{i}=U_i|Q_i|$ is the polar decomposition of $Q_{i}$ for ever $i\leq i\leq n+1$. Define
		
		$$
		f^{(z_1,\ldots, z_{n+1})}_{k_1,\ldots,k_{n+1}}(A\Phi) =\ip{Q_{n+1}\Delta_\Phi^{\iu z_n}Q_n\ldots\Delta_\Phi^{\iu z_1}Q_1\Phi}{\Delta_\Phi^{-\iu \bar{z}_{n+1}}A\Phi}_\phi.
		$$
		
		For now, we will omit the superscript index on the operators to not overload the notation.
		
		Since $\bar{f}(z_1,\ldots,z_n)\doteq\overline{f^{(z_1,\ldots, z_{n+1})}_{k_1,\ldots,k_{n+1}}(A\Phi)}$ is an analytic function, it attains its maximum at an extremal point of $S$ (see \cite{Araki73} Corollary 2.2). Denoting $z_j=x_j+\iu y_j$ and using now the Tomita-Takesaki Theorem in the similar calculation we made in equation \eqref{calculationX17.1}, we get
		
		\noindent $(i)$ for the extremal points with $\Im{z_j}=0$ for all $1\leq j\leq n+1$ we get
		\begin{equation}
		\label{calculationX18}
		\begin{aligned}
		&\left|\bar{f}(x_1,\ldots, x_{n+1})\right|\leq\\
		&\left\|A\Phi\right\| \ip{Q_{n+1}\tau^\phi_{x_n}(Q_n)\ldots\tau^\phi_{x_n+\cdots+x_1}(Q_1)\Phi}{Q_{n+1}\tau^\phi_{x_n}(Q_n)\ldots\tau^\phi_{x_n+\cdots+x_1}(Q_1)\Phi}^\frac{1}{2}\\
		& \leq\left\|A\Phi\right\|\|H\|_p^\frac{1}{2} \prod_{j=1}^{n+1}\|Q_j\|_{4(n+1)q},  \quad \forall (x_1,\ldots,x_{n+1})\in \mathbb{R}^{n+1}. \\
		\end{aligned}
		\end{equation}
		
		\noindent$(ii)$ now for $\Im{z_j}=0$ for all $i\neq l$ and $\Im{z_l}=-\frac{1}{2}$ where $l\neq n+1$
		\begin{equation}
		\label{calculationX19}
		\begin{aligned}
		&\left|\bar{f}\left(x_1,\ldots, x_l-\frac{1}{2}\iu ,\ldots, x_{n+1}\right)\right|\leq\left\|A\Phi\right\|\times\\ &\left\|\tau^\phi_{x_{n+1}}(Q_{n+1})\ldots \tau^\phi_{x_{n+1}+\cdots+x_{l+1}}(Q_{l+1})J_\Phi\tau^\phi_{x_{n+1}+\cdots+x_1}(Q_1^\ast)\ldots\tau^\phi_{x_{n+1}+\cdots+x_l}(Q_l^\ast)\Phi\right\|\\
		&\leq\left\|A\Phi\right\|\|H\|_p^\frac{1}{2}\left(\prod_{j=1}^{l}\|Q_j\|_{4lq}\right)\left(\prod_{j=l+1}^{n+1}\|Q_j\|_{4(n+1-l)q}\right) \ \forall (x_1,\ldots,x_{n+1})\in \mathbb{R}^{n+1}.\\
		\end{aligned}
		\end{equation}
		
		\noindent$(iii)$ finally, for $\Im{z_j}=0$ for all $i\neq n+1$ and $\Im{z_{n+1}}=-\frac{1}{2},$
		\begin{equation}\begin{aligned}
		\label{calculationX20}
		\left|\bar{f}\left(x_1,\ldots, x_l,\ldots,x_n, x_{n+1}-\frac{1}{2}\iu \right)\right|\\
		&\hspace{-3.5cm}=\left|\ip{Q_{n+1}\Delta_\Phi^{\iu x_n}Q_n\ldots\Delta_\Phi^{\iu x_1}Q_1\Phi}{\Delta_\Phi^{-\iu x_{n+1}}\Delta_\Phi^{\frac{1}{2}}A\Phi}_\phi\right|\\
		&\hspace{-3.5cm}=\left|\ip{\tau^\phi_{x_{n+1}}(Q_{n+1})\ldots\tau^\phi_{x_{n+1}+\cdots+x_1}(Q_1)\Phi}{J_\Phi A^\ast\Phi}_\phi\right|\\
		&\hspace{-3.5cm}\leq\left\|A\Phi\right\|\|H\|_p^\frac{1}{2} \prod_{j=1}^{n+1}\|Q_j\|_{4(n+1)q}, \forall (x_1,\ldots,x_{n+1})\in \mathbb{R}^{n+1},\\
		\end{aligned}\end{equation}
		where the last line follows by the same calculation done in equation \eqref{calculationX18}.
		
		The last step is to remember that we omitted the superscripts and notice that
		$$\begin{aligned}
		\lim_{k_{n+1}\to \infty}\ldots\lim_{k_1\to \infty} f^{(z_1,\ldots, z_{n+1})}_{k_1,\ldots,k_{n+1}}(A\Phi)& =\ip{Q_{n+1}^{k_{n+1}}\Delta_\Phi^{\iu z_n}Q_n\ldots\Delta_\Phi^{\iu z_1}Q_1^{k_1}\Phi}{\Delta_\Phi^{-\iu \bar{z}_{n+1}}A\Phi}_\phi\\
		&=\ip{Q_{n+1}\Delta_\Phi^{\iu z_n}Q_n\ldots\Delta_\Phi^{\iu z_1}Q_1\Phi}{\Delta_\Phi^{-\iu \bar{z}_{n+1}}A\Phi}_\phi,\\
		\end{aligned}$$
		so $\bar{f}(z_1,\ldots, z_{n+1})=\overline{\ip{Q_{n+1}\Delta_\Phi^{\iu z_n}Q_n\ldots\Delta_\Phi^{\iu z_1}Q_1\Phi}{\Delta_\Phi^{-\iu \bar{z}_{n+1}}A\Phi}_\phi}$ is the limit of a sequence of analytic functions uniformly bounded, thus, analytic on $S^{n+1}_\frac{1}{2}$  and bounded on its closure, as desired.
		
		As we saw in equation \eqref{calculationX20} the term $\Delta_\Phi^{z_{n+1}}$ does not interfere with the conclusion for $-\frac{1}{2}<\Im{z_{n+1}}<0$ and, by the very same argument used above to obtain a continuous linear extension of $\bar{f}(z_1,\ldots, z_{n+1})$, it follows that $Q_n\Delta_\Phi^{\iu z_n}Q_{n-1}\ldots \Delta_\Phi^{\iu z_1}Q_1\Phi \in \Dom{\Delta_\Phi^{\iu z}}$ for $-\frac{1}{2}<\Im{z}<0$.
		
	\end{proof}
	\begin{remark}
		In contrast to what we did in equation \eqref{calculationX18}, for $\Im{z_i}=0$ for all $1\leq i\leq n+1$, it holds that
		\begin{equation}
		\begin{aligned}
		\left|\bar{f}_k(z_1,\ldots, z_{n+1})\right|&=\lim_{k\to\infty}\left|\bar{f}_k(z_1,\ldots, z_{n+1})(t)\right|\\
		&=\lim_{k\to\infty}\left|\ip{Q_{n+1,k}\Delta_\Phi^{\iu x_n}Q_n\ldots\Delta_\Phi^{\iu x_1}Q_1\Phi}{\Delta_\Phi^{-\iu x_{n+1}}A\Phi}_\phi\right|\\
		&\leq \left\|\Delta_\Phi^{-\iu t}A\Phi\right\| \lim_{k\to\infty} \left\|UQ_{n+1,k}\Delta_\Phi^{\iu x_n}Q_n\ldots\Delta_\Phi^{\iu x_1}Q_1\Phi\right\|\\
		&\leq \left\|\Delta_\Phi^{-\iu t}A\Phi\right\| \lim_{k\to\infty} \left\|UQ_{n+1,k}\tau^\phi_{x_n}(Q_n)\ldots\tau^\phi_{x_1+\cdots+x_n}(Q_1)\Phi\right\|\\
		& \leq\left\|A\Phi\right\|\lim_{k\to\infty}\phi\left(UQ_{n+1,k}\tau^\phi_{x_n}(Q_n)\ldots\tau^\phi_{x_1+\cdots+x_n}(Q_1)\right)^\frac{1}{2}\\
		&\leq \left\|A\Phi\right\|\lim_{k\to\infty}\tau\left(H^\frac{1}{2} \tau^\phi_{x_1+\cdots+x_n}(Q_1^\ast)\ldots \tau^\phi_{x_n}(Q_n^\ast)Q_{n+1,k} U^\ast \right. \\ 
		&\hspace{2.8cm} \left. UQ_{n+1,k}\tau^\phi_{x_n}(Q_n)\ldots\tau^\phi_{x_1+\cdots+x_n}(Q_1) H^\frac{1}{2}\right)^\frac{1}{2}\\
		& \leq\left\|A\Phi\right\|\|H\|_p^\frac{1}{2} \lim_{k\to\infty}\left(\tau\left(|Q_{n+1,k}|^{2nq}\right)^\frac{1}{nq}\prod_{j=1}^{n}\tau\left(|Q_i|^{2nq}\right)^{\frac{1}{nq}}\right)^\frac{1}{2}\\
		& \leq\left\|A\Phi\right\|\|H\|_p^\frac{1}{2} \prod_{j=1}^{n+1}\|Q_i\|_{2nq} \quad \forall t\in \mathbb{R}.\\
		\end{aligned}
		\end{equation}
	\end{remark}
	
	\begin{corollary}
		\label{CR1}
		Let $\nalgebra\subset B(\hilbert)$ be a von Neumann algebra, $\tau$ a normal faithful semifinite trace on $\nalgebra$ and $\phi(\cdot)=\ip{\Phi}{\cdot\Phi}$ a state on $\nalgebra$. Let also $\frac{1}{p}+\frac{1}{q}=1$, $\lambda\in\overline{\mathbb{R}}_+$, $Q\in \ex^\tau_{4q,\lambda}$ and suppose $\phi=\tau_H$ is $\|\cdot\|_p$-continuous. Then, if $\Delta_\Phi^{\iu z_{j-1}}Q\ldots \Delta_\Phi^{\iu z_1}Q\Phi \in \Dom{Q}$ for every $-\frac{1}{2}\leq \Im{z_j} \leq 0$ and for every $j\in\mathbb{N}$, $$\Phi(Q)\doteq\sum_{n=0}^\infty \int_{S^n_\frac{1}{2}}dt_1\ldots dt_{n} \Delta_\Psi^{t_n}Q\Delta_\Psi^{t_{n-1}}Q\ldots \Delta_\Psi^{t_1}Q\Phi,$$ is absolutely and uniformly convergent.
	\end{corollary}
	\begin{proof}
		By Theorem \ref{TR1}, $\Delta_\Psi^{\iu z_n}Q\Delta_\Psi^{\iu z_{n-1}}Q\ldots \Delta_\Psi^{\iu z_1}Q A\Phi$ is well defined and $$\begin{aligned}
		\big\|\Delta_\Psi^{\iu z_n}Q\Delta_\Psi^{\iu z_{n-1}}&Q\ldots \Delta_\Psi^{\iu z_1}QA\Phi\big\|\\
		&\leq \left\|A\Phi\right\|\|H\|_p^\frac{1}{2}\max_{0\leq l\leq n}\left\{\underbrace{\left(\prod_{j=1}^{l}\|Q_j\|_{4lq}\right)}_{=1 \textrm{ if } l=0}\left(\prod_{j=l+1}^{n}\|Q_j\|_{4(n-l)q}\right)\right\}\\
		&=  \left\|A\Phi\right\|\|H\|_p^\frac{1}{2}\max_{0\leq l\leq n}\left\{\underbrace{\|Q\|_{4lq}^{l}}_{=1 \textrm{ if } l=0}\|Q\|_{4(n-l)q}^{n-l}\right\}\\
		&=  \left\|A\Phi\right\|\|H\|_p^\frac{1}{2}\max_{0\leq l\leq \lfloor \frac{n}{2}\rfloor}\left\{\underbrace{\|Q\|_{4lq}^{l-1}}_{=1 \textrm{ if } l=0}\|Q\|_{4(n-l)q}^{n-l}\right\}\\
		&= \left\|A\Phi\right\|\|H\|_p^\frac{1}{2}\max_{0\leq l\leq \lfloor \frac{n}{2}\rfloor}\left\{\underbrace{\tau\left(|Q|^{4lq}\right)^\frac{l}{4lq}}_{=1 \textrm{ if } l=0}\tau\left(|Q|^{4(n-l)q}\right)^\frac{n-l}{4(n-l)q}\right\}.
		\end{aligned}$$
		
		Notice that, for $\displaystyle Q_m=\int_0^m \lambda dE_\lambda^{|Q|}$, we have that $$f_m(z)\doteq\tau\left(Q_m^{4 z q}\right)^\frac{1}{4q}\tau\left(Q_m^{4(n-z)q}\right)^\frac{1}{4q}$$
		is an analytic function in the region $\{z\in \mathbb{C} \ | \ 1\leq\Re{z}\leq n-1\}$, since this region is a strip that does not intercept the negative real line. Hence its modulus in the region mentioned is assumed when $\Re{z}=1$ or $\Re{z}=n-1$ by the Maximum Modulus Principle. In these cases, we have
		$$	|f_m(1+\iu t)|= |f_m(n-1+\iu t)|=\|Q_m\|_{4q}\|Q_m\|_{4(n-1)q}^{n-1}.$$
		
		As usual, taking the limit $m\to \infty$ we obtain, for all $z\in\{w\in\mathbb{C} \ | \ 1\leq \Re{w} \leq n-1 \}$,
		\begin{equation}
		\label{Riesz-Thorin}
		\tau\left(Q^{4 z q}\right)^\frac{1}{4q}\tau\left(Q^{4(n-z)q}\right)^\frac{1}{4q}\leq \|Q\|_{4q}\|Q\|_{4(n-1)q}^{n-1}.
		\end{equation}
		
		Finally, the series
		$$\sum_{n=0}^\infty \int_{0}^{t} dt_1\ldots \int_{0}^{t_n}dt_{n} \Delta_\Psi^{t_n}Q\Delta_\Psi^{t_{n-1}}Q\ldots \Delta_\Psi^{t_1}QA\Phi$$ is $\|\cdot\|$-convergent. In fact, considering first the case $Q\in \ex^\tau_{4q,\lambda}$ with $0<\lambda<\infty$, there exists $N\in\mathbb{N}$ such that, for all $ k,l >N$, $\displaystyle \lambda \sum_{n=k}^l\frac{\lambda^{n-1}\|Q\|_{4(n-1)q}^{n-1}}{(n-1)!}<\frac{\epsilon}{2}$ and $\displaystyle \frac{\|Q\|_{4q}}{N}< 1$, thus
		
		$$\begin{aligned}
		\Bigg \|\sum_{n=k}^l \int_{0}^{t} dt_1\ldots &\int_{0}^{t_{n-1}}dt_{n} \Delta_\Psi^{ t_n}Q\Delta_\Psi^{t_{n-1}}Q\ldots \Delta_\Psi^{ t_1}QA\Phi\Bigg\|\\
		&\hspace{0.8cm}\leq \sum_{n=k}^l \int_{0}^{t} dt_1\ldots \int_{0}^{t_{n-1}}dt_{n} \left \|\Delta_\Psi^{t_n}Q\Delta_\Psi^{ t_{n-1}}Q\ldots \Delta_\Psi^{ t_1}QA\Phi\right\|\\
		&\hspace{0.8cm}\leq \sum_{n=k}^l\frac{t^n\max\left\{\|Q\|_{4q}\|Q\|_{4(n-1)q}^{n-1},\|Q\|_{4nq}^{n}\right\}}{n!}\\
		&\hspace{0.8cm}\leq \sum_{n=k}^l\frac{t^n\left(\|Q\|_{4q}\|Q\|_{4(n-1)q}^{n-1}+\|Q\|_{4nq}^{n}\right)}{n!}\\
		&\hspace{0.8cm}\leq  \sum_{n=k}^l \frac{\|Q\|_{4q}}{n}\frac{t^n\|Q\|_{4(n-1)q}^{n-1}}{(n-1)!}+\frac{t^n\|Q\|_{4nq}^{n}}{n!}\\
		&\hspace{0.8cm}<\epsilon.
		\end{aligned}$$
		
		For the case $\lambda=\infty$, just remember that $\displaystyle \ex^\tau_{p,\infty}=\bigcap_{\lambda\in\overline{ \mathbb{R}}_+}\ex^\tau_{p,\lambda}$.
	\end{proof}
	
	\begin{proposition}
		\label{convergence_regularity}
		Let $(Q_n)_{n\in\mathbb{N}} \subset\ex_{4q,\lambda}^\tau$ be a sequence such that $Q_n\xrightarrow{\|\cdot\|_{4mq}}Q\in\ex_{4q,\lambda}^\tau$, $\|Q_n\|_{4mq}\leq \|Q\|_{4mq}$ and $\|Q-Q_n\|_{4mq} \leq M$ for all $m\in\mathbb{N}$. In addition, suppose that, for each fixed $n\in\mathbb{N}$,
		$\Phi \in\Dom{Q_1}$ and $\Delta_\Phi^{\iu z_{j-1}}Q_{j-1}\ldots \Delta_\Phi^{\iu z_1}Q_1\Phi \in \Dom{Q_j}$ for every $-\frac{1}{2}\leq \Im{z_j} \leq 0$ and for every $2\leq j\leq n$, where $Q_j$ can be either $Q_n$ or $Q$. Then
		$$Exp_{l,r}\left(\int_0^t;Q_n(s) ds\right)\Phi\xrightarrow{n\to\infty} Exp_{l,r}\left(\int_0^t;Q(s) ds\right)\Phi \ , \quad t<\lambda.$$
	\end{proposition}
	\begin{proof}
		First, notice that $Q$ and $Q_n$, $n\in\mathbb{N}$, are densely defined closed operators. Furthermore, they have a common core due to the increasing hypothesis and the properties of $\tau$-dense subsets.
		
		Define $$A^m_j(z_1,\ldots, z_n)=\Delta_\Phi^{\iu z_m} Q \ldots\Delta_\Phi^{\iu z_{j-1}}Q\Delta_\Phi^{\iu z_{j}}(Q-Q_n)\Delta_\Phi^{\iu z_{j+1}}Q_n\ldots \Delta_\Phi^{\iu z_1}Q_n\Phi.$$
		
		Using a telescopic sum argument, we have, for $m>1$,
		$$\begin{aligned}
		&\left\|\Delta_\Phi^{\iu z_m} Q\ldots \Delta_\Phi^{\iu z_1}Q\Phi-\Delta_\Phi^{\iu z_n} Q_n\ldots \Delta_\Phi^{\iu z_1}Q_n\Phi\right\| &\\
		&\hspace{2.57cm}=\left\|\sum_{j=1}^m A^m_j(z_1,\ldots, z_n) \right\|\\
		&\hspace{2.57cm}\leq \sum_{j=1}^m \left\| A^m_j(z_1,\ldots, z_n) \right\|\\
		&\makebox[\textwidth]{$\displaystyle \hfill\leq \sum_{j=1}^m \|H\|_p^\frac{1}{2}\max_{0\leq l\leq m-1}\left\{\|Q\|_{4lq}^l\|Q-Q_n\|_{4(m-l)q}\|Q_n\|_{4(m-l)q}^{m-l-1}\right\}$}\\
		&\hspace{2.57cm}\leq m \, \|H\|_p^\frac{1}{2}\max_{0\leq l\leq m-1}\left\{\|Q\|_{4lq}^l\|Q\|_{4(m-l)q}^{m-l-1}\|Q-Q_n\|_{4(m-l)q}\right\}\\
		&\hspace{2.57cm}= m \, \|H\|_p^\frac{1}{2}\max_{0\leq l\leq m-1}\left\{\|Q\|_{4lq}^l\|Q\|_{4(m-l)q}^{m-l}\frac{\|Q-Q_n\|_{4(m-l)q}}{\|Q\|_{4(m-l)q}}\right\}.\\
		\end{aligned}$$
		
		Applying Equation \eqref{Riesz-Thorin} to the inequality above we get
		
		$$\begin{aligned}
		&\left\|\Delta_\Phi^{\iu z_m} Q\ldots \Delta_\Phi^{\iu z_1}Q\Phi-\Delta_\Phi^{\iu z_n} Q_n\ldots \Delta_\Phi^{\iu z_1}Q_n\Phi\right\| & \\
		&\makebox[\textwidth]{$\hfill\displaystyle \leq m \, \|H\|_p^\frac{1}{2}\|Q\|_{4q}\|Q\|_{4(m-1)q}^{m-1}\max_{0\leq l\leq m-1}\left\{\frac{\|Q-Q_n\|_{4(m-l)q}}{\|Q\|_{4(m-l)q}}\right\}.$}\\
		\end{aligned}$$
		Hence,
		\begin{equation}
		\label{calcx34}
		\begin{aligned}
		& \left\|Exp_{l,r}\left(\int_0^t;Q(s) ds\right)\Phi-Exp_{l,r}\left(\int_0^t;Q_n(s) ds\right)\Phi\right\| & \\
		&\hspace{0.3cm}\leq \|H\|_p^\frac{1}{2}\|Q-Q_n\|_{4q}+\\
		&\makebox[\textwidth]{$\displaystyle \hfill +\sum_{m=2}^{\infty}\frac{t^m}{m!} m \, \|H\|_p^\frac{1}{2}\|Q\|_{4q}\|Q\|_{4(m-1)q}^{m-1}\max_{0\leq l\leq m-1}\left\{\frac{\|Q-Q_n\|_{4(m-l)q}}{\|Q\|_{4(m-l)q}}\right\}$}\\
		&\hspace{0.3cm} = \|H\|_p^\frac{1}{2}\|Q-Q_n\|_{4q}+ \\
		& \makebox[\textwidth]{$\displaystyle \hfill + \|H\|_p^\frac{1}{2}\|Q\|_{4q} \, t \sum_{m=2}^{\infty}\frac{t^{m-1}}{(m-1)!} \|Q\|_{4(m-1)q}^{m-1}\max_{0\leq l\leq m-1}\left\{\frac{\|Q-Q_n\|_{4(m-l)q}}{\|Q\|_{4(m-l)q}}\right\}.$}\\
		\end{aligned}
		\end{equation}
		
		Finally, let $\epsilon>0$ be given. Since $Q\in \ex_{4q,\lambda}^\tau$, there exists $m_0\in\mathbb{N}$ such that, for all $m\leq m_0$, 
		$$\sum_{m=M}^{\infty}\frac{t^{m-1}}{(m-1)!} \|Q\|_{4(m-1)q}^{m-1}<\frac{\epsilon}{3M}.$$
		By hypothesis, there also exists $n_0\in\mathbb{N}$ such that 
		$$\frac{\|Q-Q_n\|_{4mq}}{\max\left\{\|Q\|_{4nq},1\right\}}<\epsilon \left[3 \lambda \|H\|_p^\frac{1}{2}  \sum_{m=2}^{\infty}\frac{t^{m-1}}{(m-1)!} \|Q\|_{4(m-1)q}^{m-1}\right]^{-1} \, , \quad \forall n\geq n_0. $$
		
		It follows from Equation \eqref{calcx34} that
		$$\left\|Exp_{l,r}\left(\int_0^t;Q(s) ds\right)\Phi-Exp_{l,r}\left(\int_0^t;Q_n(s) ds\right)\Phi\right\|<\epsilon \, , \quad \forall n\geq n_0.$$
	\end{proof}
	
	One of the consequences of Proposition \ref{convergence_regularity} is that the sequence of Araki's perturbations obtained by the upper cut in the spectral decomposition of the modulus of a $\ex_{4q,\lambda}^\tau$-perturbation converges to the perturbation described in Corollary \ref{CR1}.
	
	Moreover, Proposition \ref{convergence_regularity}  gives us an interpretation for the parameter $\lambda$ in $\ex^\tau_{p,\lambda}$.
	
	It is important to mention the notorious similarity between our approach and Sakai's geometric vectors (which were mentioned right after equation \eqref{expansional_exponential}). In the direction of equation \eqref{expansional_exponential} we can obtain the following result:
	
	\begin{corollary}
		Let $\nalgebra\subset B(\hilbert)$ be a von Neumann algebra, $\tau$ a normal faithful semifinite trace on $\nalgebra$ and $\phi(\cdot)=\ip{\Phi}{\cdot\Phi}$ a state on $\nalgebra$. Let also $\frac{1}{p}+\frac{1}{q}=1$, $\lambda\in\overline{\mathbb{R}}_+$, $Q=Q^\ast\in \ex^\tau_{4q,\lambda}$ and suppose $\phi$ is $\|\cdot\|_p$-continuous. In addition, suppose that $\Delta_\Phi^{\iu z_{j-1}}Q\ldots \Delta_\Phi^{\iu z_1}Q\Phi \in \Dom{Q}$ for every $-\frac{1}{2}\leq \Im{z_j} \leq 0$ and for every $j\in\mathbb{N}$. Then, for every $z\in \mathbb{C}$ with $0<\Re{z}<\frac{1}{2}$,  $\Phi\in \Dom{e^{z(H+Q)}}$ and
		$$e^{z(H_\Phi+Q)}\Phi=\sum_{n=0}^\infty \left(\frac{z}{2}\right)^n\int_{S^n_\frac{1}{2}}dt_1\ldots dt_{n} \Delta_\Psi^{t_n} Q \Delta_\Psi^{t_{n-1}}Q\ldots \Delta_\Psi^{t_1}Q\Phi,$$
		where $H_\Phi=\log\Delta_\Phi$.
	\end{corollary}
	\begin{proof}
		Notice that for $z=\iu t$, $t\in\mathbb{R}$, one identify the right-hand side with $Exp_l\left(\int_0^t;\tau^\Phi_t(Q) ds\right)$. Thus, for purely imaginary $z$, we have the equality.
		
		Now, since we know by Theorem \ref{CR1} that the left-hand side is analytic in the region $0< \Re{z}<\frac{1}{2}$, \cite[Proposition 4.12]{Araki73.2} guarantees the thesis.

	\end{proof}

	We can use all previous results to conclude with a general theorem. Notice that, for the KMS condition, the interesting case is the case $\lambda=\frac{1}{2}$ as one can see in the definition of the expansional, in \cite{Araki73.2}.
	
	\begin{theorem}
		Let $\nalgebra\subset B(\hilbert)$ be a von Neumann algebra, $\tau$ a normal faithful semifinite trace on $\nalgebra$, $(\nalgebra,\alpha)$ a $W^\ast$-dynamical system, $H_\Phi$ the Hamiltonian of $\alpha$  and $\phi(\cdot)=\ip{\Phi}{\cdot\Phi}$ a ($\tau$, $\beta$)-KMS state on $\nalgebra$. Let also $\frac{1}{p}+\frac{1}{q}=1$, $\lambda\in\overline{\mathbb{R}}_+$, $Q=Q^\ast\in \ex^\tau_{4q,\lambda}$ and suppose $\phi$ is $\|\cdot\|_p$-continuous. In addition, suppose that $\Delta_\Phi^{\iu z_{j-1}}Q\ldots \Delta_\Phi^{\iu z_1}Q\Phi \in \Dom{Q}$ for every $-\frac{1}{2}\leq \Im{z_j} \leq 0$ and for every $j\in\mathbb{N}$.
		Then, $\Phi\in \Dom{e^{-\frac{\beta}{2}(H_\Phi+Q)}}$ and $\phi^Q(A)=\ip{\Psi^Q}{A\Psi^Q}$ is as $(\tau^Q,\beta)$-KMS state for 	the perturbed dynamics defined by
		$\alpha_t^Q(A)=e^{\iu t(H_\omega+Q)}Ae^{-\iu t(H_\omega+Q)}$,  where $\Psi^Q=\frac{e^{-\frac{\beta}{2}(H_\omega+Q)}\Omega_\omega}{\|e^{-\frac{\beta}{2}(H_\omega+Q)}\Omega_\omega\|}$.
	\end{theorem}
	
	We would add at this point that, in contrast to Araki's treatment in \cite{Araki73}, we do not believe our results prove any kind of stability of KMS states. This belief is based on the necessity to add a ``dual'' continuity property on the state, \ie, we need an additional $\|\cdot\|_p$-continuity hypothesis with index $p$ H\"older-conjugated with the one used to control the perturbation. Hence, all the stability we proved seems to be a consequence of that continuity. An important exception is the stability of the domain of the Modular Operator, which allows us to extend the multiple-time KMS condition to unbounded operators, proved in Theorems \ref{TR0} and \ref{TR1}.
	
	\if\typeiii1
	We can naturally extend the results presented hitherto to general von Neumann algebras using Haagerup's construction. In fact, Lemma \ref{identificationofweights} states that we can consider a state $\phi$ on $\nalgebra$ as a state $\tilde{\phi}$ on $\nalgebra_\alpha\rtimes \mathbb{R}$ satisfying
	$$\psi \circ \theta_t=\psi \quad \forall t\in \mathbb{R}. $$. Since we identify $\nalgebra$ with $\pi_\alpha(\nalgebra)\subset\nalgebra_\alpha\rtimes \mathbb{R}$, the spectral decomposition guarantees that any operator affiliated with $\nalgebra$ can be seen as an operator affiliated with $\nalgebra_\alpha\rtimes \mathbb{R}$.
	
	Now, there exists a trace $\tau$ on $\nalgebra_\alpha\rtimes \mathbb{R}$
	
	\fi
	
	\if\typeiii0
	
	Although all the applications of Noncommutative $L_p$-Spaces to Physics know by the author are restricted to semifinite von Neumann algebras, see, for example \cite{germinet2005}, \cite{Bru2017}, \cite{germinet2011}, and \cite{nittis2016}, we finish saying that the author is aware of the limitations imposed by the existence of a faithful normal semifinite trace on the algebra.
	
	Several results have been proved about the type of the algebras in relativistic AQFT in the past decades, showing that, under some physical reasonably assumptions, the algebra of observables of a diamond has to be of type III, see \cite[Section V.6]{haag2012} and \cite[Proposition 3.2]{Buchholz87}.
	
	Fortunately for general von Neumann algebras, either in the Haagerup or in the Araki-Masuda construction, there is a natural trace related with the noncommutative $L_p$-space. This suggests that our ideas can be generalized. In fact, roughly speaking, Haagerup's generalization of noncommutative $L_p$-spaces for general von Neumann algebras uses several identifications between a von Neumann algebra (and other objects related to it) and the crossed product $\nalgebra\rtimes_{\tau^\varphi} \mathbb{R}$, where $\tau^\varphi=\{\tau^{\varphi}_t\}_{t\in\mathbb{R}}$ is the modular automorphism group obtained throughout the faithful normal and semifinite weight $\varphi$. Among these identifications we highlight: (i) the one due to \cite[II, Lemma 1]{terp81}, which says that the mapping $\phi\mapsto \tilde{\phi}$, where $\tilde{\phi}=\hat{\phi}\circ T$ and $\hat{\phi}$ is a natural extension of $\phi$ to $\widehat{\nalgebra}_+$ (the extended positive part of $\nalgebra$) as described in \cite[Proposition 1.10]{Haagerup79I}, is a bijection from the set of all normal semifinite weights on $\nalgebra$ onto the set of all normal semifinite weights $\psi$ on $\nalgebra\rtimes_{\tau^\varphi}\mathbb{R}$ satisfying $\psi \circ \theta_t=\psi \quad \forall t\in \mathbb{R}$; (ii) the one due to \cite[Theorem 1.2]{haagerup79}, which says that
	$$\begin{aligned}
	&\{H_\phi \in \left(\nalgebra\rtimes_{\tau^\varphi} \mathbb{R}\right)_\eta \, | \, \phi \textrm{ is normal and semifinite} \}\\
	&\makebox[\textwidth]{$\displaystyle\hfill =\{H_\phi \in \left(\nalgebra\rtimes_{\tau^\varphi} \mathbb{R}\right)_\eta \, | \, \theta_t H_\phi=e^{-t}H_\phi \},$}\\
	\end{aligned}$$
	where $H_\phi\eta \left(\nalgebra\rtimes_{\tau^\varphi} \mathbb{R}\right)$ is the Radon-Nikodym derivative for the normal semifinite weight $\phi$ on $\nalgebra$ with respect to the trace $\tau$, \ie, $H_\phi$ is the operator affiliated with $\nalgebra\rtimes_{\tau^\varphi} \mathbb{R}$ such that $\tilde{\phi}=\tau_{H_\phi}$. That means that we can consider the state $\phi$ as a state in $\nalgebra\rtimes_{\tau^\varphi}\mathbb{R}$.
	
	It is important to notice that $\nalgebra\rtimes_{\tau^\varphi}\mathbb{R}$ has a natural trace $\tau$, \cite[Lemma 5.2]{Haagerup79II}, which is used in the definition of the noncommutative $L_p$-space, namely, 
	$$\begin{aligned}
	L_p(\nalgebra)&\doteq\left\{H\in\left(\nalgebra\rtimes_{\tau^\varphi} \mathbb{R}\right)_\tau \ \middle| \ \theta_t H=e^{-\frac{t}{p}}H, \ \forall t \in \mathbb{R} \right\}\\
	L_\infty(\nalgebra)&\doteq\left\{H\in\left(\nalgebra\rtimes_{\tau^\varphi} \mathbb{R}\right)_\tau \ \middle| \ \theta_t H=H, \ \forall t \in \mathbb{R} \right\}.\\\end{aligned}$$
	
	The trace $\tau$ is not used to define the norm on these spaces, instead a positive linear function, $tr(\cdot)$ is defined on $L_1(\nalgebra)$. This linear function satisfies a trace-like property $tr(AB)=tr(BA)$, where $A\in L_p(\nalgebra)$ and $B\in L_q(\nalgebra)$ and ${\frac{1}{p}+\frac{1}{q}=1}$.
	
	In addition, it also easy to see that there is a natural inclusion of the operators affiliated with $\nalgebra$ in the operators affiliated with $\nalgebra\rtimes_{\tau^\varphi}\mathbb{R}$, because of the spectral decomposition.
	
	Again, we stress that, if in one hand the development presented here is not appropriated to deal with perturbations of KMS states in Relativist Algebraic Quantum Field Theory, it is a natural framework to deal with Statistical Mechanics, Linear Response and Information Theory.

	\fi
	

\printbibliography
	
\end{document}